\documentclass[10pt]{article}

\usepackage{graphicx} % Required for inserting images
\usepackage{cite}             % improves presentation of citations
\usepackage{amssymb}
\usepackage{amsmath}
\usepackage[margin=1in]{geometry}
\usepackage{amsthm}
\usepackage{authblk}
\usepackage{comment}
\usepackage{cancel}
\usepackage{color, soul}
\usepackage{overpic}
\usepackage{hyperref}
\usepackage{scalerel,stackengine}

    \newtheorem{proposition}{Proposition}
    \newtheorem{lemma}{Lemma}
    \newtheorem{theorem}{Theorem}
    
    \newtheorem{rem}{Theorem}
    \newtheorem{lem}{Lemma}
    \newtheorem{definition}{Definition}
    \newtheorem{claim}{Claim}

\begin{document}
    
    \title{Round-Preserving Asymptotic Compression of Prior-Free Interactive Protocols} 
    
    %%%%%%
    
    \author{Gurleen Padda}
    \author{Dave Touchette}
    \affil{Department of Computer Science \& Institut Quantique, Université de Sherbrooke, Sherbrooke, QC, Canada\\Email: \{gurleen.padda, dave.touchette\}@usherbrooke.ca}
    \date{}
    \maketitle
    %%%%%

          There is a close relationship between the communication complexity and information complexity of communication problems, as demonstrated by results such as Shannon’s noiseless source coding theorem, and the Slepian–Wolf theorem. Here, we study this relationship in the prior-free and interactive setting, where we provide an alternate proof for the result of Braverman [SIAM Review, vol. 59, no. 4, 2017], that the amortized communication complexity of simulating a prior-free interactive communication protocol, is equal to its prior-free information cost. While this is a known result, our approach addresses the need for a more natural proof of it. We also improve on the result by achieving round preservation, and using a bounded quantity of shared randomness. We do this by showing that the communicating parties can produce a reliable estimate of the joint type, or empirical distribution, of their inputs. This estimate is then used in our protocol for the prior-free reverse Shannon theorem with side information at the receiver. These results are then generalized to the interactive setting to obtain our main result.
    
    \section{Introduction} \label{sec:introduction}
        \subsection{Motivation and Prior Work}
            \subsubsection{One-way communication}
                \paragraph{Source coding}
                    The relationship between communication and uncertainty was first formalized by Shannon's source coding theorem \cite{shannon}. It captures the notion that the more uncertain a receiver is about the message they could be receiving, the more communication the sender will need to use. The Shannon entropy $H$ quantifies this uncertainty: for a random variable $X\sim\mu_{X}$, and it's alphabet of possible outcomes $\mathcal{X}$, the Shannon entropy is defined as
                    \begin{equation}
                        H(X)_{\mu} = -\sum_{x \in \mathcal{X}} \mu_{X}(x)\log(\mu_{X}(x)).
                    \end{equation}
                    The lower the entropy $H(X)_{\mu}$, the further the distribution $\mu_{X}$ is from that of a uniform distribution, and vice-versa. Shannon's source coding theorem establishes that in order to communicate a sequence of $n$ symbols, denoted $x^{n}$, where each symbol of $x^{n}=x_{1}, x_{2},...,x_{n}$ is an identically and independently distributed (i.i.d.) instance of the random variable $X$, the optimal bound on the number of bits that a sender could expect to communicate to the receiver is $nH(X)_{\mu}$. In this case, the prior distribution $\mu_{X}$ is known to both the sender and the receiver. We refer to the number of bits that need to be communicated to complete a task as the task's communication complexity. It is often convenient to assume the existence of a referee who generates the sequence $x^{n}$ and gives it to the sender, who is then tasked with sending the sequence to the receiver.
                    
                \paragraph{Side information}
                    The relationship between communication and uncertainty carries over into a slightly modified setting, where we introduce another random variable $Y$. This, along with the random variable $X$, is described by a joint distribution $X, Y \sim \mu_{XY}$,  with its alphabet of possible outcomes $\mathcal{X}\times\mathcal{Y}$. This distribution is known to both the sender and the receiver. In addition to giving a sequence $x^{n}$ to the sender, the referee also provides the receiver with a sequence $y^{n}=y_{1}, y_{2},...,y_{n}$, where these sequences consist of $n$ i.i.d. instances of $X, Y$. The sequence $y^{n}$, often referred to as the receiver's side information, may therefore be correlated with the sequence $x^{n}$ that the sender would like to communicate to them. In this setting, the receiver's uncertainty about what an instance of $X$ could be is quantified by the conditional entropy $H(X|Y)_{\mu}$, which is defined as
                    \begin{equation}
                        H(X|Y)_{\mu} = -\sum_{(x, y) \in \mathcal{X} \times \mathcal{Y}} \mu_{XY}(x, y)\log(\mu_{X|Y}(x|y)).
                    \end{equation}
                    The Slepian-Wolf coding bound establishes that in order for the sender to communicate $x^{n}$ to a receiver that already holds $y^{n}$, an optimal bound on the communication complexity is $nH(X|Y)_{\mu}$ bits \cite{slepian}. If $X$ and $Y$ are correlated, then the conditional entropy $H(X|Y)_{\mu}$ is smaller than the marginal entropy $H(X)_{\mu}$. The correlation between $X$ and $Y$ can also be quantified through the mutual information 
                    \begin{equation}
                        I(X;Y)_{\mu} = H(X)_{\mu}-H(X|Y)_{\mu}.
                    \end{equation}
                    The stronger the correlation between $X$ and $Y$, the lower the conditional entropy $H(X|Y)_{\mu}$, and the higher the mutual information $I(X;Y)_{\mu}$. 
                    
                \paragraph{Message simulation}
                    Moving to a different setting, we now equip the sender with the ability to generate a new sequence $m^{n}=m_{1}, m_{2},...,m_{n}$ from their original sequence $x^{n}$ through $n$ i.i.d. uses of a noisy channel $N$. The output of this noisy channel is associated with a random variable $M$, an output distribution $p(m|x)$, and an output alphabet $\mathcal{M}$. This noisy channel can be applied independently on each of the $n$ symbols of $x^{n}$ to obtain $n$ new symbols that form the sequence $m^{n}$. If such a noisy communication channel existed between the sender and the receiver, in particular one that could return a copy of the channel output to the sender, the sender could input each symbol of $x^{n}$ into the channel one at a time, and the sender and receiver would both end up with a copy of the sequence $m^{n}$. We refer to a channel that returns a copy of the output to the sender as a feedback channel. We will use $N^{n}$ to denote the resulting product channel that generates $m^{n}$ from $x^{n}$ according to the output distribution $p^{n}(m^{n}|x^{n})=p(m_{1}|x_{1})\cdot p(m_{2}|x_{2})...p(m_{n}|x_{n})$. 
                    
                    In general, we will use $\pi$ to refer to a communication protocol, and $\pi^{n}$ to refer to an asymptotic communication protocol, where the players are interested in executing $n$ independent instances of the protocol $\pi$. In the setting of the reverse Shannon theorem with feedback \cite{bennett2002}, $\pi$ is the protocol where the sender inputs a symbol $x$ to the noisy channel $N$, and the sender and receiver both obtain the output $m$. $\pi^{n}$ is then the protocol where the noisy channel is applied identically and independently on $n$ i.i.d. input symbols $x^{n}$ to obtain $n$ output symbols $m^{n}$. In this setting, the sender and receiver wish to use a noiseless bit channel and some shared randomness to simulate the protocol $\pi^{n}$. The sender is given an input sequence $x^{n}$ that is an i.i.d. sequence distributed according to $\mu_{X}$, where $\mu_{X}$ is known to both the sender and the receiver. The sender and receiver wish to use the noiseless channel and shared randomness to both end up with a copy of the same sequence $\hat{m}^{n}$ that is correctly correlated with the sender's original input sequence $x^{n}$ according to the distribution of the noisy channel $p^{n}$ that is known to both the sender and the receiver. 
                    
                    The protocol that simulates $\pi^{n}$ will be denoted $\hat{\pi}^{n}$, and the communication complexity of the simulation $\hat{\pi}^{n}$, which is the number of noiseless bits communicated from sender to receiver, is denoted $CC(\hat{\pi}^{n})$. We will use $SR(\hat{\pi}^{n})$ to denote the number of shared random bits used. The goal of the simulation is to maintain a low communication complexity, or compress the original protocol, while ensuring that the output distribution of the simulated noisy channel is close in total variation distance (see Ref.\cite{pollard2002user} for a definition) to that of the original channel $N$. We give a formal definition for simulating interactive communication protocols in Definition~\ref{def:sim_int}. The simulation protocol $\hat{\pi}^{n}$ should achieve the same task as the protocol $\pi^{n}$, the difference between the two being the resources available to the players to complete the task. In the case where the receiver holds no side information, there exists a protocol $\hat{\pi}^{n}$ that simulates $\pi^{n}$, and achieves the optimal communication complexity of $nI(M;X)_{\mu, p}$ bits. The players also use $nH(M|X)_{\mu, p}$ shared random bits to achieve this complexity\cite{bennett2002, bennett_devetak_harrow_shor_winter_2014}. 
                    
                    The reverse Shannon theorem was also studied in the non-asymptotic, or one-shot setting, in Ref.\cite{hjmr}, where the sender and receiver wish to simulate the protocol $\pi$, which is one use of the noisy channel. In this case, the optimal bound on the communication complexity is shown to be $I(M;X)_{\mu, p}$ bits on average for variable length output messages.
                    
                \paragraph{Message simulation with side information}
                    The communication cost of the reverse Shannon theorem without side information captures the notion that the receiver's uncertainty about what $\hat{m}^{n}$ should be is related to their uncertainty about the information it contains about $x^{n}$. This is quantified by the mutual information $I(M;X)_{\mu, p}$. We now consider the case where the receiver is equipped with some side information $y^{n}$, where $x^{n}$ and $y^{n}$ consist of $n$ i.i.d. instances of $X, Y\sim \mu_{XY}$, and  $\mu_{XY}$ is known to both the sender and the receiver. They again wish to use a noiseless bit channel and some shared randomness to simulate the protocol $\pi^{n}$. The message $\hat{m}^{n}$ that is generated by the simulation $\hat{\pi}^{n}$ should be correctly correlated with the sender's input $x^{n}$ through the distribution $p^{n}$. In this case, any correlation between $M$ and $Y$ exists only through $X$, and $p^{n}(m^{n}|x^{n})=p^{n}(m^{n}|x^{n}, y^{n})$. The communication cost $CC(\hat{\pi}^{n})$ of simulating $\pi^{n}$ becomes $nI(M;X|Y)_{\mu, p}$ \cite{braverman_rao}, where $I(M;X|Y)_{\mu, p}$ is the conditional mutual information,
                    \begin{equation}
                        I(M;X|Y)_{\mu, p} = H(M|Y)_{\mu, p}-H(M|X,Y)_{\mu, p}.
                    \end{equation}
                    The simulation protocol of Ref.\cite{braverman_rao} is an extension of a protocol that works in the one-shot setting, which achieves the optimal bound on the communication complexity of $I(M;X|Y)_{\mu, p}$. A drawback of their result, however is that the simulation protocol requires that the sender and receiver exchange messages back and forth across the noiseless channel in order to simulate using the noisy channel in one direction from sender to receiver. They therefore have to use an interactive protocol, which we define in the next section, to simulate one-way communication, even when the prior distribution is known to the sender and receiver. Their protocol also requires an unbounded quantity of shared randomness.

            \subsubsection{Interactive communication}
            \label{subsubsec:int_comm}
                In each of the cases that we have touched on so far, there has been a distinct sender and receiver. The information moves in one direction from the sender to the receiver, and we therefore refer to this as one-way communication. We will now move to the interactive setting, where there are two players that will have to communicate and send information back and forth to each other.

                \paragraph{Interactive communication protocols}
                    Just as in the case of message simulation with side information, we assume that there is a referee who provides the communicating parties Alice and Bob each with an $n$ symbol i.i.d. sequence $x^{n}$ and $y^{n}$ respectively, which are sampled from a known distribution $\mu_{XY}$. They are also given the description of $j$ noisy feedback channels $N_{1},...,N_{j}$. An interactive communication protocol $\pi^{n}$ is iterated by rounds, and at each round of communication, one noisy channel is used. For round $i$, channel $N_{i}$ is used. A new round begins when the direction of communication changes, i.e. when Alice goes from being the sender to the receiver, and vice-versa for Bob. In this model, for channels with odd indices, Alice is the sender and Bob is the receiver, and vice versa for even numbered channels. At each round, the player who is acting as the sender will input their original sequence (either $x^{n}$ or $y^{n}$) and all the messages $m^{n}_{1},...,m^{n}_{i-1}$ that were generated in all previous rounds into the product channel $N^{n}_{i}$ to obtain a new message $m^{n}_{i}$. Channel $N_{i}$ has as an output alphabet $\mathcal{M}_{i}$ and either $\mathcal{X} \times \mathcal{M}_{1}\times ... \times \mathcal{M}_{i-1}$ or $\mathcal{Y} \times \mathcal{M}_{1}\times ... \times \mathcal{M}_{i-1}$ as an input alphabet, depending on the parity of the channel's index. Each channel is associated with an output distribution of the form $p_{i} = p_{M_{i}|XM_{1}...M_{i-1}}$ or $p_{i} = p_{M_{i}|YM_{1}...M_{i-1}}$ depending on the parity of $i$. Throughout the paper, when referring to tuples of the form $\mathcal{M}_{1} \times \mathcal{M}_{2} \times \dots \times \mathcal{M}_{j}$, we will use the shorthand $\mathcal{M}_{\leq j}$, or $\mathcal{M}_{< j}$ when $\mathcal{M}_{j}$ is excluded. Similarly, for expressions such as $(m_1, m_2, \dots, m_j)$ or $M_{1}M_{2} \dots M_{j}$, we will write $m_{\leq j}$ and $M_{\leq j}$ respectively, or $m_{< j}$ and $M_{< j}$ when $m_j$ and $M_{j}$ respectively are excluded. We may also refer to the distribution $p_{1}\cdot p_{2}\cdot...\cdot p_{j}$ as $p_{\leq j}$.
                    
                    The $n$ symbols of the outputted sequence $m^{n}_{i}$ will be denoted $m_{i1},m_{i2},...,m_{in}$, where the first index of each symbol corresponds to the round in which the sequence was generated, and the second to the symbol's position in the sequence. Assuming $i$ is odd, an output of the product channel $m_{i}^{n}$ is sampled from $M_{i}^{n}=N^{n}_{i}\left(x^{n},m^{n}_{<i}\right)$, which is distributed according to the product distribution $p_{i}^{n}\left(m^{n}_{i}|x^{n},m^{n}_{<i}\right)$, defined as
                    \begin{equation}
                        M_{i}^{n} \sim p_{i}\left(m_{i1}|x_{1},m_{<i1}\right)
                        ... p_{i}\left(m_{in}|x_{n},m_{<in}\right).
                    \end{equation}
                    After $i$ rounds, the players would have generated the following transcript of $i$ messages
                    \begin{align}
                        &\label{eqn:first_round} M_{1}^{n} \sim N_{1}^{n}(x^{n})
                        \\&M_{2}^{n} \sim N_{2}^{n}(y^{n}, m^{n}_{1}) 
                        \\&...
                        \\&\label{eqn:mult_rounds}M_{i}^{n} \sim N_{i}^{n}(x^{n}, m^{n}_{<i}).
                    \end{align}
                    $\Pi(X, Y)$ or just $\Pi$ denotes the random variable of the transcript when $x$ and $y$ are also random variables. 
            The probability associated with the transcript $\pi^{n}(x^{n}, y^{n})$ is given by
                \begin{equation}
                    p^{n}(m^{n}_{\leq j}|x^{n}, y^{n})= p^{n}_{j}(m^{n}_{j}|x^{n}, m^{n}_{<j})\cdot p^{n}_{j-1}(m^{n}_{j-1}|y^{n}, m^{n}_{<(j-1)})\cdot...\cdot p^{n}_{1}(m^{n}_{1}|x^{n}).
                \end{equation}
                To generate their final outputs from the protocol, Alice and Bob are each equipped with a function $f_{A}$ and $f_{B}$ respectively, that each takes the players input $x^{n}$ or $y^{n}$, the protocol transcript $\pi^{n}(x^{n}, y^{n})$, all of the local randomness $S_{A}$ and $S_{B}$ respectively, and the shared randomness $R$ used throughout the entire protocol. At the end of the protocol, the players each output
                \begin{align}
                    f_{A}\left(x^{n}, \pi(x^{n}, y^{n}), S_{A}, R\right)\\
                    f_{B}\left(y^{n}, \pi(x^{n}, y^{n}), S_{B}, R \right)
                \end{align}
                respectively.
                
                    A more detailed presentation of interactive communication protocols can be found in Section~\ref{sec:int_prelims}, where this protocol is modeled as a rooted tree of depth $j$, as per Ref.~\cite{Braverman2017}. A transcript of the protocol is obtained by traversing the tree along a path from the root to a leaf. A visual representation is also provided in Figure~\ref{fig:int_mod} in Section~\ref{sec:int_prelims}.
                        
                \paragraph{Simulating interactive communication}
                    In order to simulate a $j$ round interactive communication protocol $\pi^{n}$, we assume that Alice and Bob have access to two noiseless bit channels, one from Alice to Bob and the other from Bob to Alice, as well as some shared randomness. The communication complexity is the total number of noiseless bits exchanged across both channels over all rounds of communication. Their goal is to both end up with the same transcript of $j$ messages $\hat{m}_{1}^{n},...,\hat{m}_{j}^{n}$, such that the joint output distribution of the $j$ simulated channels is close in total variation distance to that of the original channels, while minimizing the communication complexity. This was studied in Ref.~\cite{braverman_rao}, where it is shown that the communication complexity of asymptotically simulating an interactive communication protocol for i.i.d. inputs $x^{n}$ and $y^{n}$ from a known prior distribution $\mu_{XY}$ is equivalent to $nIC_{\mu}(\pi)$, where $IC_{\mu}(\pi)$ is the information cost of the protocol, and that this is an optimal bound on the communication complexity. The information cost of an interactive protocol $\pi$ quantifies how much information is revealed to each player about the other's input through the execution of the protocol $\pi$ on inputs from a distribution $\mu_{XY}$. It is defined as 
                        \begin{equation}
                            IC_{\mu}(\pi) = I(\Pi;X|Y)_{\mu, p_{\leq j}}+I(\Pi;Y|X)_{\mu, p_{\leq j}},
                        \end{equation}
                    where $\Pi$ is a random variable associated with the transcript outputted by the protocol. In the model for interactive communication described above, $\Pi = M_{1},...,M_{j}$. The simulation protocol presented in Ref.~\cite{braverman_rao} may require multiple rounds of communication, where Alice and Bob exchange messages back and forth, in order to simulate one round of the original interactive protocol. We note that simulating an interactive communication protocol that has $1$ round of communication reduces to the reverse Shannon theorem with side information.

                \paragraph{Distributional versus prior-free communication} 
                    So far, all of the settings that we have considered assume the existence of a prior distribution $\mu$, from which the referee samples the inputs that they will give to each player. This is known as the distributional communication setting. In the prior-free setting, there is no distribution $\mu$ associated with the inputs. The simulation protocol has to work for the ``worst" possible input, but the channel is still applied identically and independently on each symbol of the inputted sequence. The players have no knowledge of how the inputs were generated. 
                    
                    The prior-free reverse Shannon theorem without side information was studied in \cite{bennett2002, bennett_devetak_harrow_shor_winter_2014}, where the optimal bound on the communication complexity is shown to be the maximum mutual information $nI(M;X)_{\mu, p}$, where the maximum is taken over all possible prior distributions $\mu_{X}$. This is summarized in the following theorem. 
                    \begin{rem}[Theorem 1 of \cite{bennett_devetak_harrow_shor_winter_2014}]
                    Given an alphabet $\mathcal{X}$, let $\pi^{n}$ be the protocol where Alice inputs $x^{n} \in \mathcal{X}^{n}$ into the noisy feedback channel $N^{n}$, that is distributed according to $p^{n}_{M|X}(m^{n}|x^{n})$, and her and Bob both obtain a copy of the output $m^{n}$. There exists a simulation protocol $\hat{\pi}^{n}$ such that
                        \begin{equation}
                            \lim_{n\rightarrow \infty}\frac{CC(\hat{\pi}^{n})}{n} \leq \max_{\mu_{X}}I(M;X)_{\mu_{X}, p}.
                        \end{equation}
                    \end{rem}
                    This is the optimal communication complexity to simulate one round of communication in the prior-free setting, where the sender does not make use of the receiver's side information. For interactive protocols, we define the prior-free information complexity below.
                \begin{definition}[Prior-free information complexity]
                    The prior-free information complexity $IC(\pi)$ is defined as 
                    \begin{equation}
                            IC(\pi) = \max_{\mu} IC_{\mu}(\pi).
                    \end{equation}
                    \end{definition}
                    Prior-free interactive communication was studied in Ref.~\cite{Braverman2017}, where it is shown that the prior-free information complexity of an interactive protocol $\pi$ is equivalent to the amortized prior-free communication complexity of the protocol $\hat{\pi}^{n}$. This is summarized in the following theorem, for which we provide an alternate proof for with round preservation, and a bound on the amount of shared randomness used in Theorem~\ref{thm:int_2}, the main result of our paper.
                    \begin{rem}[Theorem 6.7 of \cite{Braverman2017}]
                        For every interactive protocol $\pi^{n}$, there exists a simulation protocol $\hat{\pi}^{n}$, such that
                    \begin{equation}
                       \lim_{n\rightarrow \infty}\frac{CC(\hat{\pi}^{n})}{n} \leq IC(\pi).
                    \end{equation}
                    \end{rem}
                     This result is demonstrated in Ref.~\cite{Braverman2017} by extending a protocol that operates in the one-shot setting, that is based on iterated correlated sampling, to the asymptotic setting. This is then combined with a minimax theorem, to obtain a bound on the prior-free communication complexity from the distributional communication complexity for the worst case possible input distribution. The simulation protocol given in Ref.~\cite{Braverman2017} requires multiple rounds of communication to simulate a single round of the original protocol, and it therefore lacks round-preservation. Their protocol also requires an unbounded quantity of shared randomness.
            
        \subsection{Summary of Results}
            In this paper we improve on the results of Ref.~\cite{Braverman2017} with a simulation protocol that relies on the theory of types to obtain this result. The method of types often provides a more intuitive and straightforward approach to analyzing information theory problems, especially in the context of source coding and data compression, and with it we are able to obtain a more natural proof for the compression results of Ref.~\cite{Braverman2017}. It is well-suited for analyzing the asymptotic behaviour of random processes and sequences. This does limit the applicability of our findings to the asymptotic setting, unlike some of the results of Ref.~\cite{Braverman2017} which are effective in the one-shot setting as well. On the other hand, an important distinction between our results and those of Ref.~\cite{Braverman2017} is that our simulation protocol preserves the number of rounds of the original protocol. In addition, we are able to obtain an upper bound on the quantity of shared randomness needed while maintaining the desired communication cost.
            
            \paragraph{Round preservation}
                 Round complexity can be used as a complexity measure for certain distributed information-processing tasks. It captures the dependence that interaction has on making optimal use of communication resources. In Ref.~\cite{Braverman2017}, the number of rounds of communication needed to simulate a single round of the original interactive communication protocol is unbounded. This was improved in Ref.~\cite{braverman2013direct}, where they assess the trade-off between increasing the overall communication complexity, and decreasing the round complexity of the simulation protocol. In particular, they show that it is possible to simulate $j$ rounds of communication using a $7j$ round protocol. In Ref.\cite{jain_2012}, they study interactive communication in the bounded-round setting, where there is an upper bound on the number of rounds that the protocol they wish to simulate can have. They use the same compression technique of Ref.~\cite{braverman_rao} for one-way communication protocols and extend it to the interactive setting to obtain a direct product result.
                 
                 In our results, we develop a new compression technique for prior-free one-way communication (Theorems~\ref{thm:rst_1}-\ref{thm:rst_2}) which, when extended to the prior-free and interactive setting, shows that any prior-free $j\-$~round interactive protocol $\pi$ can be simulated by a $(j+1)\-$~round protocol $\hat{\pi}^{n}$, while still achieving the optimal communication complexity (Theorem~\ref{thm:int_2}). We also show that the simulation can be performed in $j$ rounds at the expense of increasing the communication cost of the first round of the simulation (Theorem~\ref{thm:int_3}). In the single round setting, for the special case where the noisy channel $N$ is the identity channel, we show that our protocol achieves the Slepian-Wolf coding bound, and uses no shared randomness (Theorems~\ref{thm:sw_1}-\ref{thm:sw_2}). We also study the case where the noisy channel $N$ is still the identity channel, but only Alice's input is prior-free, and Bob's input is generated from Alice's input through the use of a noisy channel (Theorem~\ref{thm:sw_3}).
                 
            \subsubsection{Techniques}
                \paragraph{Theory of types}
                     The approach of Ref.~\cite{bennett_devetak_harrow_shor_winter_2014} for the distributional reverse Shannon theorem in the setting without side information, is to use a random binning strategy to uniformly partition the set of potential messages into a predefined number of bins, such that the partitioning is known to the sender and the receiver. Since the prior distribution of the input is known, the receiver can determine the set of messages that are most likely to occur using typical sets. 
                    In the prior-free reverse Shannon theorem of Ref.~\cite{bennett_devetak_harrow_shor_winter_2014}, the theory of types is used in order to substitute prior distributions with empirical distributions. When the receiver has no side information, it is sufficient for the sender to determine the empirical distribution of their input locally, which they can then communicate to the receiver using a sub-linear (logarithmic) amount of communication, since the number of possible types for a sequence from a fixed alphabet is at most polynomial in the sequence length. 

                \paragraph{Joint empirical distribution estimation}
                    Since we are working in the setting where the receiver has side information, the results of \cite{bennett_devetak_harrow_shor_winter_2014} cannot be applied directly, as the marginal empirical distribution of the sender's input sequence does not make use of the receiver's side information. In order to make use of the side information, the players would need the joint empirical distribution of their sequences. This poses a problem, as neither the sender nor the receiver has any information about the other's sequence to be able to determine their joint empirical distribution. 
                    
                    Our approach is to allow Alice and Bob to learn just enough about the other player's input to be able to reliably estimate the joint empirical distribution of their sequences. We then show that this estimate is reliable enough to be used to compress simulated communication protocols with side information. We show that it is possible to do this while maintaining the optimal communication complexity of the simulation. 

            \subsubsection{Results}
                We give a protocol that allows Alice and Bob to obtain a reliable estimate of the joint empirical distribution $t_{x^{n}, y^{n}}$ of their inputs by using some shared randomness and a sublinear amount of communication to uniformly sample and exchange coordinates along their sequences. This is summarized in the following lemma.
                \begin{lem}[Informal summary of Lemma~\ref{lem:estimating}]
                    Let $n,m \in \mathbb{N}$, with $m < n$. There exists a 2-round protocol that uses $m \log\left(|\mathcal{X} | \cdot |\mathcal{Y}| \right)$ bits of noiseless communication and $m \log\left(n\right)$ bits of shared randomness, such that for any input $(x^{n}, y^{n}) \in \mathcal{X}^{n} \times \mathcal{Y}^{n}$, Alice and Bob both obtain an estimate $\tilde{t}$ that is close in total variation distance to the true empirical distribution $t_{x^{n}, y^{n}}$, with a probability of success that is exponentially close to $1$.
                \end{lem}
                We then use this estimate to simulate communication protocols in the prior-free setting. We show that given such an estimate, for any noisy channel there exists a protocol that simulates $n$ channel uses from Alice to Bob, such that the output distribution of the simulated channel is close in total variation to that of the true channel. This is summarized in the following theorem.
                \begin{rem}[Informal summary of Theorem~\ref{thm:rst_1}]
                    There exists a $1$-round protocol $\hat{\pi}^{n}$ taking $(x^{n}, y^{n}) \in \mathcal{X}^{n}\times \mathcal{Y}^{n}$ as inputs, a description of the noisy channel $N$ associated with the output distribution $p_{M|X}$, and a joint type $\widetilde{t}$ on the alphabet $\mathcal{X}\times \mathcal{Y}$, with a guarantee that $\widetilde{t}$ is close in total variation distance to the true joint type $t_{x^{n}, y^{n}}$ of Alice and Bob's input sequences, such that
                    \begin{align}
                        \lim_{n\rightarrow \infty} \frac{CC(\hat{\pi}^{n})}{n} &\leq I\left(M;X|Y\right)_{t\cdot p}\\
                        \lim_{n\rightarrow \infty} \frac{SR(\hat{\pi}^{n})}{n} &\leq H\left(M|X, Y\right)_{t\cdot p}
                    \end{align}
                    where we've used the shorthand $t\cdot p = t_{x^{n}, y^{n}}\cdot p_{M|X}$. It holds with a probability that is exponentially close to $1$ that Alice and Bob each produce a copy of the same output $\hat{m}^{n}$, such that the likelihood of obtaining this output from the simulation protocol $\hat{\pi}^{n}$ given $x^{n}$ is close in total variation distance to $p_{M|X}$.
                \end{rem}
                As mentioned earlier, in the special case that the channel $N$ is the identity channel, our results comply with the Slepian-Wolf coding bound (Theorems~\ref{thm:sw_1}-\ref{thm:sw_3}). The protocol that achieves these results has the players arrange the possible channel outputs in uniformly random order in a list. They then sample some shared randomness to choose a subset of messages in the list, to obtain a reduced list, within which the sender chooses a message that is correctly correlated with their input according to the output distribution of the channel. The sender then communicates the marginal type of the chosen message using $O(\log n)$ bits of communication, and then enough bits of the binary expansion of the message's position in the reduced list so that the receiver can identify the message that the sender has chosen.

                This protocol requires the use of shared randomness in two different instances, once to arrange the channel outputs in uniformly random order, and a second time to choose a subset of messages in the list. We follow the proof of Newman's theorem \cite{newman_1991} given in Ref.~\cite{Rao_Yehudayoff_2020}, to reduce the amount of shared randomness used in the the first instance from an exponential quantity in $n$ to a linear quantity, and show that guarantees similar to that of the original protocol can still be obtained. We then have one player locally sample and communicate this randomness to the other player using $O(\log n)$ bits of communication, rather than sampling from the shared randomness. The second instance in which shared randomness is used remains untouched, as this is what allows the players to offset the communication cost of simulating the channel.
                
                The results of Theorem~\ref{thm:rst_1} are then combined with Lemma~\ref{lem:estimating} in Theorem~\ref{thm:rst_2}, and subsequently extended to the interactive setting in Theorem~\ref{thm:int_2}. Applying the de-randomization technique described above \cite{newman_1991, Rao_Yehudayoff_2020}, we obtain the following novel result for interactive protocols.
                
                \begin{rem}[Informal summary of Theorem~\ref{thm:int_2}]
                    There exists a $j+1$-round protocol $\hat{\pi}^{n}$ taking $\left(x^{n}, y^{n}\right)$ as inputs,  as well as a description of $j$ noisy channels $N_{1}, ..., N_{j}$, each associated with an output distribution $p_{M_{i}|X, M_{< i}}$ or $p_{M_{i}|Y, M_{< i}}$ depending on the parity of $i$, such that Alice and Bob can, with a probability that is exponentially close to $1$, reliably simulate $n$ executions of a $j$ round interactive protocol $\pi$, and share a copy of the transcript of the simulated protocol, such that 
                    \begin{align}
                        \lim_{n\rightarrow \infty} \frac{CC(\hat{\pi}^{n})}{n} &\leq I\left(M_{\leq j};X|Y\right)_{t\cdot p}+I\left(M_{\leq j};Y|X\right)_{t\cdot p}\\
                        \lim_{n\rightarrow \infty} \frac{SR(\hat{\pi}^{n})}{n} &\leq H\left(M_{\leq j}|X, Y\right)_{t\cdot p}
                    \end{align}
                    where we've used the shorthand $t\cdot p = t_{x^{n}, y^{n}}\cdot p_{\leq j}.$
                \end{rem}
                
                We then show in Theorem \ref{thm:int_3} that the simulation protocol can be performed in $j$ rounds, at the expense of increasing the communication cost of the first round of the simulation.
                
                The paper is divided as follows. Section \ref{sec:estimating} contains the empirical distribution estimation protocol and a brief discussion of the typical sets of the estimated distribution. Section \ref{sec:single} contains protocols for the prior-free compression of one way communication protocols, with prior-free Slepian-Wolf coding in Section \ref{sec:sw} and the prior-free reverse Shannon theorem in Section \ref{sec:rst}. We then move to the compression of prior-free interactive protocols in Section \ref{sec:interactive}.

%%%%%%%%%%%%%%

    \section{Notation and Preliminaries}                \label{sec:int_prelims}
        \subsection{Information theory}

            \begin{definition}[Shannon entropies]
                Let $\mu_{X}$ be a distribution on the alphabet $\mathcal{X}$. The entropy of the random variable $X \sim \mu_{X}\left(x\right)$ is defined as 
            \begin{equation}
                H\left(X\right)_{\mu} = -\sum_{x \in \mathcal{X}}\mu_{X}\left(x\right)\log_{2}\left(\mu_{X}\left(x\right)\right).
            \end{equation}
            Let $\mu_{X, Y}$ be a joint distribution on the alphabet $\mathcal{X} \times \mathcal{Y}$. The entropy of the random variables ${X, Y \sim \mu_{X, Y}\left(x, y\right)}$ is defined as 
            \begin{equation}
                H\left(X, Y\right)_{\mu} = -\sum_{x, y \in \mathcal{X} \times \mathcal{Y}}\mu_{X, Y}\left(x, y\right)\log_{2}\left(\mu_{X, Y}\left(x, y\right)\right).
            \end{equation}
            The conditional entropy is defined as 
            \begin{equation}
                H\left(X|Y\right)_{\mu} = -\sum_{x, y \in \mathcal{X} \times \mathcal{Y}}\mu_{X, Y}\left(x, y\right)\log_{2}\left(\mu_{X|Y}\left(x|y\right)\right).
            \end{equation}
            The joint entropy $H\left(X, Y\right)_{\mu}$ and conditional entropy $H\left(X|Y\right)_{\mu}$ are related by the marginal entropy $H\left(Y\right)_{\mu}$ in the following way
            \begin{equation}
                H\left(X|Y\right)_{\mu} = H\left(X, Y\right)_{\mu} - H\left(Y\right)_{\mu}.
            \end{equation}
            \end{definition}
            \begin{definition}[Mutual information]
            The mutual information between two random variables ${X, Y \sim \mu_{X, Y}\left(x, y\right)}$ is defined as
            \begin{equation}
                I\left(X ; Y\right) = H\left(X\right)_{\mu} - H\left(X|Y\right)_{\mu}.
            \end{equation}
            The mutual information is symmetric in it's inputs,
            \begin{equation}
                I(X;Y)=I(Y;X),
            \end{equation}
                which implies that 
            \begin{equation}
                I(X;Y)= H\left(Y\right)_{\mu} - H\left(Y|X\right)_{\mu}.
            \end{equation}
            Let ${M, X, Y \sim \mu_{M, X, Y}}$ be a set of jointly distributed random variables. The conditional mutual information is defined as
            \begin{equation}
                I\left(M;X|Y\right)_{\mu} = H\left(M|Y\right)_{\mu}-H\left(M|X, Y\right)_{\mu}.
            \end{equation}
            \end{definition}
            \begin{definition}[Relative entropy]
                Let $\mu_{X}$ and $\overline{\mu}_{X}$ be probability distributions on the alphabet $\mathcal{X}$. The relative entropy or Kullback-Leibler divergence of the distributions is defined as
                \begin{equation}
                    D(\mu||\overline{\mu}) = \sum_{x \in \mathcal{X}}\mu_{X}(x)\log_{2}\left(\frac{\mu_{X}(x)}{\overline{\mu}_{X}(x)}\right).
                \end{equation}
            \end{definition}
            \begin{lemma}[Chain rules]
                Let $M_{1}, M_{2}, X, Y \sim p_{M_{1}, M_{2}, X, Y} $ be random variables. The following chain rule holds for their mutual information
                \begin{equation}
                    I(M_{1}M_{2};X|Y)_{p} = I(M_{1};X|Y)_{p}+I(M_{2};X|YM_{1})_{p}.
                \end{equation}
\begin{comment}
    Let $\mu_{X, Y}$ and $\overline{\mu}_{X, Y}$ be probability distributions on the alphabet $\mathcal{X}\times \mathcal{Y}$. The following chain rule holds for their relative entropy
    \begin{equation}
        D(\mu_{X, Y}||\overline{\mu}_{X, Y}) = D\left(\mu_{X}||\overline{\mu}_{X}\right) +\mathbb{E}_{\mu_{X}}\left[D\left(\mu_{Y|X}||\overline{\mu}_{Y|X}\right)\right] .
    \end{equation}
\end{comment}    
            \end{lemma}

            \begin{definition}[Total variation distance]
                Let $\mu_{X}$ and $\overline{\mu}_{X}$ be probability distributions on the alphabet $\mathcal{X}$. The total variation distance of the distributions, defined in terms of the $\ell^{1}$ norm is
                \begin{align}
                    \Delta(\mu, \overline{\mu}) &= \frac{1}{2}\big\|\mu_{X}-\overline{\mu}_{X}\big\|_{1}\\
                    &=\frac{1}{2}\sum_{x \in \mathcal{X}}\Big|\mu_{X}(x)-\overline{\mu}_{X}(x)\Big|.
                \end{align}
            \end{definition}
\begin{comment}
\begin{lemma}[Data processing inequalities]
                    Let $X, Y, Z$ be three random variables that form a Markov chain $X \rightarrow Y \rightarrow Z$, then it holds that $I(X; Y) \geq I(X, Z)$ \cite{cover_thomas_2006}.

                    \cite{Levin_Peres_Wilmer_Propp_Wilson_2017}
                \end{lemma}
\end{comment}
                \begin{lemma}[Continuity bounds \cite{Winter_2016}]\label{lem:cont_bounds}
                    For distributions $\mu_{X, Y}$ and $\overline{\mu}_{X, Y}$ on the alphabet $\mathcal{X}\times \mathcal{Y}$, Pinsker's inequality relates the relative entropy and the $\ell^{1}$ norm as follows
                    \begin{equation}
                        D(\mu||\overline{\mu})\geq \frac{1}{2\ln 2}\big\|\mu-\overline{\mu}\big\|_{1}^{2}.
                    \end{equation}
                    If $\Delta(\mu, \overline{\mu}) \leq \delta$, then by the Fannes inequality
                    \begin{equation}
                        \left|H\left(X, Y\right)_{\mu} - H\left(X, Y\right)_{\overline{\mu}}\right| \leq \gamma\left( \left|\mathcal{X}\right|\cdot \left| \mathcal{Y}\right|-1, \delta \right),
                    \end{equation}
                    and by the Alicki-Fannes inequality
                    \begin{equation} \label{eqn:af}
                        \left|H\left(X|Y\right)_{\mu} - H\left(X|Y\right)_{\overline{\mu}}\right| \leq \gamma\left( \left|\mathcal{X}\right|, \delta \right),
                    \end{equation}
                    where $\gamma\left(d, \delta \right) = \delta\log_{2}\left(d\right)+h_{2}\left(\delta\right)$, and $h_{2}$ denotes the binary entropy function.
                \end{lemma}

       \subsection{Types and empirical distributions}
            \label{sec:typ_sets}

                \begin{definition}[Types and empirical distributions\cite{cover_thomas_2006}]
                Let $x^{n} \in \mathcal{X}^{n}$ and $y^{n}\in \mathcal{Y}^{n}$ be sequences consisting of $n$ symbols from the alphabets $\mathcal{X}$ and $\mathcal{Y}$ respectively.
                    Let $f\left(x|x^{n}\right)$ denote the number of occurrences of the symbol $x \in \mathcal{X}$ in the sequence $x^{n} \in \mathcal{X}^{n}$. The type of the sequence $x^{n}$ is the empirical distribution $t_{x^{n}}$, defined as
                    \begin{equation}
                        t_{x^{n}}\left(x\right) = \frac{f\left(x|x^{n}\right)}{n}.
                    \end{equation}
                    We can also define the joint empirical distribution of the sequences, where $f\left(x, y|x^{n}, y^{n}\right)$ denotes the number of occurrences of the pair $(x, y)\in \mathcal{X} \times \mathcal{Y}$ in the pair of sequences $(x^{n}, y^{n}) \in \mathcal{X}^{n} \times \mathcal{Y}^{n}$. The joint type of $x^{n}, y^{n}$ is the empirical distribution $t_{x^{n}, y^{n}}$, defined as
                    \begin{equation}
                        t_{x^{n}, y^{n}}\left(x, y\right) = \frac{f\left(x, y|x^{n}, y^{n} \right)}{n}.
                    \end{equation}
                    Note that
                    \begin{equation}
                        t_{x^{n}}\left(x\right) = \sum_{y \in \mathcal{Y}}t_{x^{n}, y^{n}}\left(x, y\right).
                    \end{equation}
                    For $t_{y^{n}}(y) > 0$, we will also define the conditional type $t_{x^{n}|y^{n}}$ as
                    \begin{equation}
                        t_{x^{n}|y^{n}}(x|y)=\frac{t_{x^{n}, y^{n}}\left(x, y\right)}{ t_{y^{n}}\left(y\right)}.
                    \end{equation}
                \end{definition}
                \begin{definition}[Type classes \cite{cover_thomas_2006}]
                Let $\mathcal{T}^{\mathcal{X}^{n}}$ be the set of all possible types for sequences in $\mathcal{X}^{n}$. For $\widetilde{t} \in \mathcal{T}^{\mathcal{X}^{n}}$, the type class $T^{X^{n}}_{\widetilde{t}}$ is defined as the set of sequences in $X^{n}$ that have the type $\widetilde{t}$
                    \begin{equation}
                    T^{X^{n}}_{\widetilde{t}} = \left\{x^{n} \in \mathcal{X}^{n}:\forall x  \in \mathcal{X}, t_{x^{n}}\left(x\right)=\widetilde{t}\left(x\right)\right\}.
                \end{equation}
                Let $\mathcal{T}^{\mathcal{X}^{n}, \mathcal{Y}^{n}}$ be the set of all possible types for sequences in $\mathcal{X}^{n}\times \mathcal{Y}^{n}$. We can define the joint type class $T^{X^{n}, Y^{n}}_{\widetilde{t}}$ for $\widetilde{t} \in \mathcal{T}^{\mathcal{X}^{n}, \mathcal{Y}^{n}}$ as
                \begin{align}
                    T^{X^{n},Y^{n}}_{\widetilde{t}}
                    = \big\{(x^{n}, y^{n}) \in \mathcal{X}^{n}\times\mathcal{Y}^{n} :\forall (x, y)  \in \mathcal{X} \times \mathcal{Y},  t_{x^{n}, y^{n}}\left(x, y\right)=\widetilde{t}\left(x, y\right)\big\}.
                \end{align}
                \end{definition}
                \begin{proposition}
                    The number of possible types for sequences in $\mathcal{X}^{n}$ is bound by
                    \begin{equation}
                        \left|\mathcal{T}^{\mathcal{X}^{n}}\right| = \binom{n+|\mathcal{X}|-1}{|\mathcal{X}|-1}\leq (n+1)^{|\mathcal{X}|}.
                    \end{equation}
                    Similarly for sequences in $\mathcal{X}^{n} \times \mathcal{Y}^{n} $
                      \begin{equation}
                        \left|\mathcal{T}^{\mathcal{X}^{n}, \mathcal{Y}^{n}}\right| = \binom{n+|\mathcal{X}|\cdot|\mathcal{Y}|-1}{|\mathcal{X}|\cdot|\mathcal{Y}|-1}\leq (n+1)^{|\mathcal{X}|\cdot|\mathcal{Y}|}.
                    \end{equation}
                \end{proposition}
                \begin{proposition}
                \label{prop:card_typ_class}
                The cardinality of $T^{X^{n}}_{\widetilde{t}}$ is bound as follows
                    \begin{equation}
                        (n+1)^{-\left|\mathcal{X}\right|}2^{nH(X)_{\widetilde{t}}} \leq \left|T^{X^{n}}_{\widetilde{t}}\right| \leq 2^{nH(X)_{\widetilde{t}}}.
                    \end{equation}
                    Similarly for $T^{X^{n},Y^{n}}_{\widetilde{t}}$
                    \begin{equation}
                        (n+1)^{-\left|\mathcal{X}\right|\cdot\left|\mathcal{Y}\right|}2^{nH(X, Y)_{\widetilde{t}}} \leq \left|T^{X^{n},Y^{n}}_{\widetilde{t}}\right| \leq 2^{nH(X, Y)_{\widetilde{t}}}.
                    \end{equation}
                \end{proposition}
                \begin{definition}[Typical sets \cite{cover_thomas_2006, bennett_devetak_harrow_shor_winter_2014}]
                For $\delta \geq 0$ and $\widetilde{t}$ a distribution on $\mathcal{X}^{n}, \mathcal{Y}^{n}$, we define the typical set of sequences $T^{X^{n}, Y^{n}}_{\widetilde{t}, \delta}$ as
                    \begin{equation}
                        T^{X^{n}, Y^{n}}_{\widetilde{t}, \delta} = \bigcup_{\substack{t \in \mathcal{T}^{\mathcal{X}^{n}, \mathcal{Y}^{n}}\\ \|t-\widetilde{t}\|_{1} \leq \delta}}T^{X^{n},Y^{n}}_{t}.
                    \end{equation}
                \end{definition}
                \begin{proposition}
                    For $\delta \geq 0$ and $\widetilde{t} \in \mathcal{T}^{\mathcal{X}^{n}, \mathcal{Y}^{n}}$,
                \begin{equation}
                    \widetilde{t}^{n}\left(T^{X^{n}, Y^{n}}_{\widetilde{t}, \delta}\right) \geq 1-(n+1)^{|\mathcal{X}|\cdot|\mathcal{Y}|}2^{-\frac{n\delta^{2}}{2\ln 2}}.
                \end{equation}
                \end{proposition}
                        \begin{proposition}
                        \label{prop:card_joint_typ_set}
                        For $\delta \geq 0$ and $\widetilde{t} \in \mathcal{T}^{\mathcal{X}^{n}, \mathcal{Y}^{n}}$,
                            \begin{align}
                                2^{n\left(H(X, Y)_{\widetilde{t}}-\gamma\left( \left|\mathcal{X}\right|\cdot \left| \mathcal{Y}\right|-1, \delta \right)\right)-\left|\mathcal{X}\right|\cdot \left|\mathcal{Y}\right|\log_{2}\left(n+1\right)} \leq \left|T^{X^{n}, Y^{n}}_{\widetilde{t}, \delta} \right|
                                \\\leq  2^{n\left(H(X, Y)_{\widetilde{t}}+\gamma\left( \left|\mathcal{X}\right|\cdot \left| \mathcal{Y}\right|-1, \delta \right)\right)+\left|\mathcal{X}\right|\cdot \left|\mathcal{Y}\right|\log_{2}\left(n+1\right)}.
                            \end{align}
                        \end{proposition}
                        A complete proof of Proposition~\ref{prop:card_joint_typ_set} is given in the \hyperref[proof:card_joint_typ_set]{Appendix}.

       \subsection{Interactive protocols and information cost}
            \begin{figure*}
            \center
            \includegraphics[width=\textwidth]{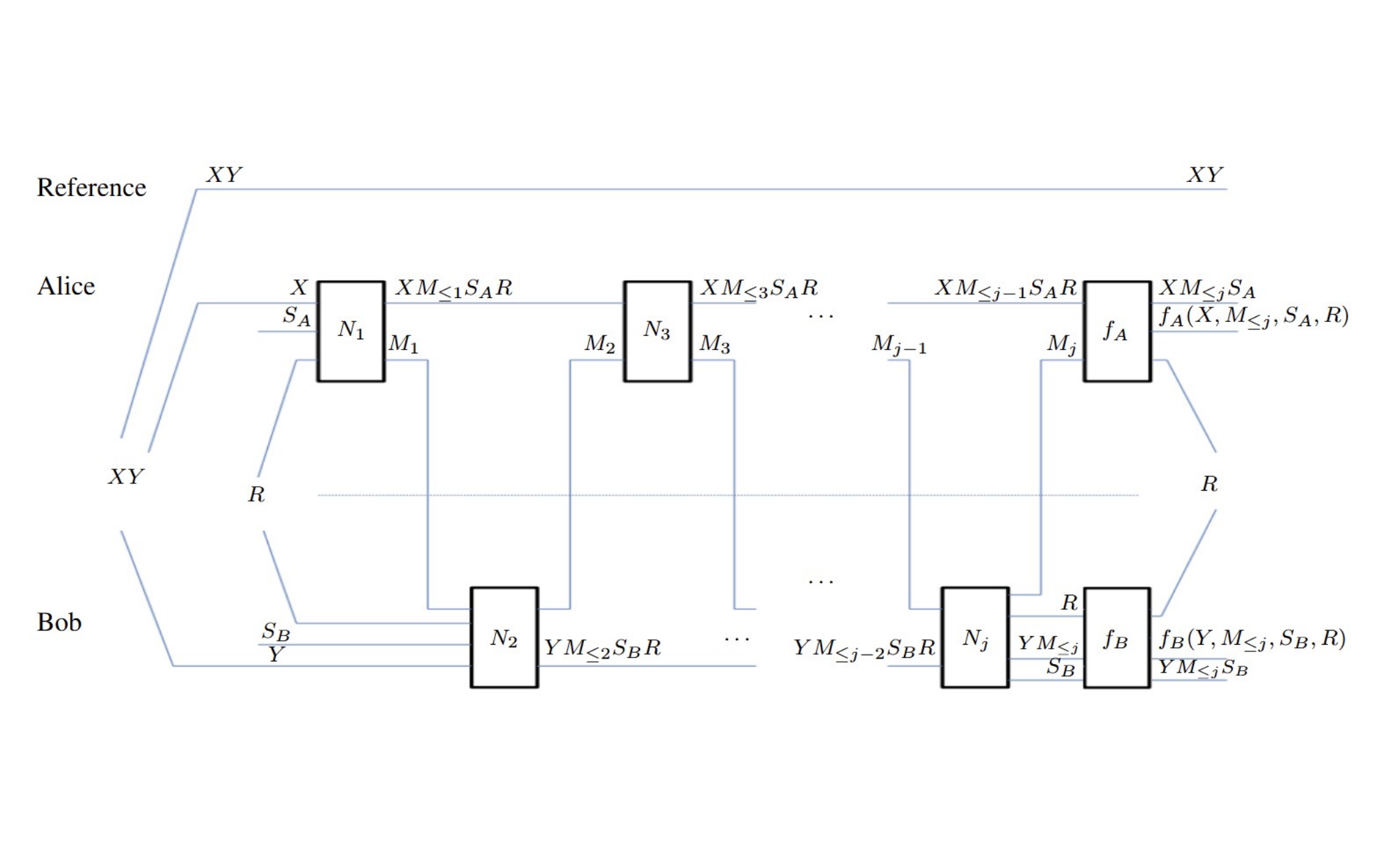}
            \vspace{-10pt}
              \caption{Depiction of a protocol in the interactive model, adapted from the long version of~\cite[Figure 1]{Tou14a} to the particular case of interactive classical communication. Alice's input register is $X$, while Bob's input register is $Y$. They have access to pre-shared randomness $R$, as well as private randomness $S_A$ and $S_B$, respectively. Alice's output registers are her input $X$ together with the message transcript $M_{\leq j} = M_1 M_2 \cdots M_j$ along with the shared ($R$) and private ($S_A$) randomness, while Bob's output registers are his input $Y$ together with the message transcript $M_{\leq j} = M_1 M_2 \cdots M_j$ along with the shared ($R$) and private ($S_B$) randomness. A copy of the joint input $XY$ of Alice and Bob is maintained in  an inaccessible reference register. The $M_i$ registers are communicated between Alice and Bob in each round, who maintain local copies upon preparations and receptions of the messages in each round. Messages $M_i$ are arbitrary apart from the fact that they are locally prepared on Alice's side for odd $i$ and on Bob's side for even $i$. Depicted here for an even number of rounds $j$, hence with Bob sending the final message; a final operation,  on Alice and Bob's side ($f_{A}$ and $f_{B}$ respectively) generates their final outputs to the protocol.}
              \label{fig:int_mod}
            \end{figure*}

%===
                \begin{definition}[Interactive communication protocols \cite{braverman_rao}]
                    Let $\mathcal{X} \times \mathcal{Y}$ denote the set of possible pairs of inputs given to two players Alice and Bob. An interactive communication protocol $\pi$ can be viewed as a rooted tree with the following structure:
                    \begin{itemize}
                        \item Each non-leaf node of the tree is owned by either Alice or Bob.
                        \item All of the children of each non-leaf node, that is owned by a given player, are owned by the other player.
                        \item A prefix free labeling is used to label each child node in such a way that no child has a label that is a prefix of another child.
                        \item Each node is associated with a function mapping $\mathcal{X}$ to distributions on children of the node and a function mapping $\mathcal{Y}$ to distributions on the children of the node.
                        \item The leaves of the tree are labeled by output values.
                        \item On input $(x, y) \in \mathcal{X} \times \mathcal{Y}$:
                        \begin{enumerate}
                            \item Let $v$ be the root of the protocol tree.
                            \item If $v$ is a leaf node, the protocol ends and the value of $v$ is outputted. Otherwise the player who owns $v$ samples a child of $v$ according to the distribution associated with their input for $v$ and sends the label of the node to the other player to indicate which child was sampled.
                            \item Set $v$ to be the newly sampled node and return to the previous step.
                        \end{enumerate}
                    \end{itemize}
                \end{definition}

            Let $\pi(x, y)$ denote all of the messages exchanged as well as the shared randomness that was sampled during the execution of the protocol $\pi$ on inputs $(x, y)$. This is the transcript of the protocol. $\Pi(X, Y)$ or just $\Pi$ denotes the random variable of the transcript when $x$ and $y$ are also random variables. Here we will interpret the model given by Ref.\cite{braverman_rao} for interactive communication as communication across noisy channels with feedback, as depicted in Figure~\ref{fig:int_mod}.
            \begin{definition}[Protocol transcript and outputs]
                Let $\pi^{n}$ denote the $j$ round interactive protocol resulting from executing $n$ identical and independent instances of the protocol $\pi$ on inputs $(x^{n}, y^{n}) \in \mathcal{X}^{n} \times \mathcal{Y}^{n}$, where
                \begin{align}
                    x^{n} &= x_{1},x_{2},...,x_{n},\\
                    y^{n} &= y_{1},y_{2},...,y_{n}.
                \end{align}
                The transcript $\pi^{n}(x^{n}, y^{n})$ consists of the following messages
                \begin{align}
                    m^{n}_{1} &= m_{11}, m_{12},...,m_{1n},\\
                    m^{n}_{2} &= m_{21}, m_{22},...,m_{2n},\\
                    ...\\
                    m^{n}_{j} &= m_{j1}, m_{j2},...,m_{jn}.
                \end{align}
            Let $p$ denote the probability of obtaining the transcript $\Pi$ from the protocol $\pi$, the probability associated with the transcript of $\pi^{n}$ is defined as
                \begin{align}
                    p^{n}(m^{n}_{1}, &m^{n}_{2}, ...,m^{n}_{j}|x^{n}, y^{n})\\
                    &=p(m_{11}, m_{21}, ...,m_{j1}|x_{1}, y_{1})
                    \cdot p(m_{12}, m_{22}, ...,m_{j2}|x_{2}, y_{2})
                    \cdot...\cdot p(m_{1n}, m_{2n}, ...,m_{jn}|x_{n}, y_{n}).
                \end{align}
                Alice and Bob are each equipped with a function $f_{A}$ and $f_{B}$ respectively, that each takes the players input $x^{n}$ or $y^{n}$, the protocol transcript $\pi^{n}(x^{n}, y^{n})$, and all of the local randomness $S_{A}$ and $S_{B}$ respectively and shared randomness $R$ used throughout the entire protocol, to generate the players final outputs. At the end of the protocol, the players each output
                \begin{align}
                    f_{A}\left(x^{n}, \pi(x^{n}, y^{n}), S_{A}, R\right)\\
                    f_{B}\left(y^{n}, \pi(x^{n}, y^{n}), S_{B}, R \right)
                \end{align}
                respectively.
            \end{definition}
            \begin{definition}[Simulating interactive protocols]
                    \label{def:sim_int}
                        Let $\pi^{n}$ be an interactive communication protocol that takes inputs from $\mathcal{X}^{n} \times \mathcal{Y}^{n}$, and fix a parameter $\epsilon>0$. A protocol $\hat{\pi}^{n}$ is said to simulate the protocol $\pi^{n}$ with error $\epsilon$ if for every pair of inputs $(x^{n}, y^{n}) \in \mathcal{X}^{n}\times\mathcal{Y}^{n}$, the output of the protocol $\hat{\pi}^{n}$ is $\epsilon$-close in total variation distance to the transcript of the protocol $\pi^{n}$. That is, if $p^{n}$ denotes the probability of obtaining transcript $m^{n}_{\leq j}$ from the protocol $\pi^{n}$ and $\widetilde{p^{n}}$ denotes the probability of obtaining $m^{n}_{\leq j}$ as an output from $\hat{\pi}^{n}$, then for every $(x^{n}, y^{n}) \in \mathcal{X}^{n}\times\mathcal{Y}^{n}$,
                \begin{equation}
                    \left\|p^{n}(m^{n}_{\leq j}|x^{n}, y^{n})-\widetilde{p^{n}}(m^{n}_{\leq j}|x^{n}, y^{n})\right\|_{1} \leq \epsilon.
                \end{equation}
            \end{definition}
                \begin{definition}[Information complexity]
                    Given alphabets $\mathcal{X}$ and $\mathcal{Y}$, a distribution $\mu_{X, Y}$ on $\mathcal{X} \times \mathcal{Y}$, and a communication protocol $\pi$ that takes inputs from $\mathcal{X} \times \mathcal{Y}$ sampled according to $\mu_{X, Y}$, the information complexity of $\pi$ is defined as 
                    \begin{equation}
                        IC_{\mu}(\pi) = I(\Pi;X|Y)+I(\Pi;Y|X).
                    \end{equation}
                    The prior-free information complexity of $\pi$ is defined as
                    \begin{equation}
                        IC(\pi) = \max_{\mu} IC_{\mu}(\pi).
                    \end{equation}
                \end{definition}
                \begin{definition}[Amortized communication complexity]
                \label{def:acc}
                    Let $CC_{\mu}(\pi)$ denote the maximum number of bits that can be exchanged during the execution of the protocol $\pi$ on inputs sampled from $\mu$, and let
                    \begin{equation}
                        CC(\pi) = \max_{\mu}CC_{\mu}(\pi).
                    \end{equation}

                    Given alphabets $\mathcal{X}$ and $\mathcal{Y}$, let $\pi^{n}$ be an interactive communication protocol that takes inputs from $\mathcal{X}^{n} \times \mathcal{Y}^{n}$. Let $\hat{\pi}^{n}$ denote a protocol that simulates the protocol $\pi^{n}$ with an asymptotically vanishing probability of error. The prior-free, amortized communication complexity of simulating $\pi^{n}$ is
                    \begin{equation}
                        \lim_{n\rightarrow \infty}\min_{\hat{\pi}^{n}}\frac{CC(\hat{\pi}^{n})}{n}.
                    \end{equation}
                \end{definition}

       \subsection{Concentration inequalities}

                \begin{lemma}[Hoeffding's inequality \cite{hoeffding_1963}] \label{lem:set_inequality}
                    Let $U$ and $V$ be random subsets of a set $Q$ of size $q$ where $|U| = u$, and $|V| = v$ and $0 < u \leq v < n$. Then $\mu = \mathbb{E}\left [ \left|U \cap V \right| \right] = \frac{uv}{q}$, and $\Pr\left[\left| U \cap V \right|\leq (1-\epsilon) \mu \right] \leq e^{-\frac{\mu \epsilon^{2}}{2}}.$
                \end{lemma}

                \begin{lemma}[Chernoff bound \cite{hoeffding_1963}]
                    Let $X_{1},...X_{n}$ be independent Poisson trials such that ${0\leq X_{i} \leq 1}$ for $i=1, ..., n$, and let $X = \frac{1}{n}\sum_{i=1}^{n}X_{i}$ and $\mu = \mathbb{E}[X]$. Then for $0<\delta< 1-\mu$, it holds that $\Pr\left[\left|X-\mu\right| \geq \delta\right] \leq e^{-2n\delta^{2}}$.
                \end{lemma}

    \section{Joint Type Estimation} \label{sec:estimating}
    \subsection{Estimating the Joint Empirical Distribution}

            In this section we will show that Alice and Bob can obtain a reliable estimate of the joint empirical distribution of their sequences $x^{n}$ and $y^{n}$. We show that there is a high probability that the estimate will be close in total variation distance to the true empirical distribution of the sequences.
            \begin{lemma} \label{lem:estimating}
                Let $\mathcal{X}$ and $\mathcal{Y}$ be finite alphabets, and let $n,m \in \mathbb{N}$, with $m < n$. For each $\delta > 0$,  there exists a 2-round protocol that uses $m \log\left(|\mathcal{X} | \cdot |\mathcal{Y}| \right)$ bits of noiseless communication and $m \log\left(n\right)$ bits of shared randomness, such that for any input $(x^{n}, y^{n}) \in \mathcal{X}^{n} \times \mathcal{Y}^{n}$, Alice and Bob both obtain an estimate $\tilde{t}$ of the true empirical distribution $t_{x^{n}, y^{n}}$, and
                \begin{equation}\label{eqn:t_bar_prob}
                    \Pr\left[\forall (x, y) \in \mathcal{X} \times \mathcal{Y}:\left| \widetilde{t}\left(x, y\right)-  t_{x^{n}, y^{n}}\left(x, y\right) \right| \leq \delta\right]
                    \geq 1- 2^{-2m\delta^{2}\log e+\log|\mathcal{X}|\cdot|\mathcal{Y}|+1}.
                \end{equation}
            \end{lemma}
            \begin{proof}
                The following protocol will produce an estimate $\tilde{t}$ of $t_{x^{n}, y^{n}}$:
                \begin{enumerate}
                    \item Alice and Bob use the shared randomness to sample $m$ coordinates with replacement along their sequences, and they each determine the string formed by the symbols found at these coordinates.
                    \item Alice and Bob exchange these strings using $m \log\left(|\mathcal{X} | \cdot |\mathcal{Y}| \right)$ bits of noiseless communication.
                    \item Let $c\left(x, y\right)$ denote the number of times the pair $\left(x, y\right)$ occurs in the exchanged sequence symbols. The estimated empirical distribution is then 
                    \begin{equation}\label{eqn:t_bar_def}
                        \begin{array}{lr}
                            \tilde{t}\left(x, y\right) =\frac{c\left(x, y\right)}{m}.
                        \end{array}
                    \end{equation}
                \end{enumerate}  
                
                The analysis for the claimed bounds follows through a standard Chernoff bound argument, details of the proof can be found in the \hyperref[proof:estimating]{Appendix}.
             \end{proof}

    \subsection{Typical Sets of the Estimated Empirical Distribution}

In this section, we will develop the machinery of the method of types which will allow us to work effectively with an estimated type $\widetilde{t} \in \mathcal{T}^{\mathcal{X}^{n}, \mathcal{Y}^{n}}$, as opposed to a known type, for sequences in $\mathcal{X}^{n} \times \mathcal{Y}^{n}$.

\subsubsection{Conditional type classes}
             \begin{definition}\label{def:cond_typ_class}
                     The type class $T^{X^{n}\backslash Y^{n}}_{\widetilde{t}}$ is defined as
                    \begin{equation}
                        T^{X^{n}\backslash Y^{n}}_{\widetilde{t}} = \Big\{x^{n}\in \mathcal{X}^{n}: \exists y^{n} \in \mathcal{Y}^{n} \mbox{ s.t. } (x^{n}, y^{n}) \in T^{X^{n}, Y^{n}}_{\widetilde{t}} \Big\}.
                    \end{equation}
                    \end{definition}
                    \begin{claim} 
                    \label{clm:marg_cond_typ_class}
                    The set $T^{X^{n}\backslash Y^{n}}_{\widetilde{t}}$ is equivalent to the marginal type class $T^{X^{n}}_{\widetilde{t}}$
                    \begin{equation}
                        T^{X^{n}\backslash Y^{n}}_{\widetilde{t}} = T^{X^{n}}_{\widetilde{t}}.
                    \end{equation}
                    \end{claim}
                    A complete proof of Claim~\ref{clm:marg_cond_typ_class} is given in the \hyperref[proof:marg_cond_typ_class]{Appendix}.
                    \begin{definition}
                     For $y^{n} \in T^{Y^{n}}_{\widetilde{t}}$, the conditional type class $ T^{X^{n}| y^{n}}_{\widetilde{t}}$ is defined as
                    \begin{equation}
                        T^{X^{n}| y^{n}}_{\widetilde{t}} = \Big\{x^{n}\in \mathcal{X}^{n} :  (x^{n}, y^{n}) \in T^{X^{n},Y^{n}}_{\widetilde{t}} \Big\}.
                    \end{equation}
                \end{definition}
                \begin{claim} The joint type class $T_{\widetilde{t}} ^{X^{n}, Y^{n}}$ can be partitioned into disjoint subsets that each form a conditional type class
                \label{clm:cond_typ_class}
                    \begin{equation}
                        T_{\widetilde{t}} ^{X^{n}, Y^{n}} = \bigcup_{y^{n} \in T^{Y^{n}}_{\widetilde{t}}}T^{X^{n}| y^{n}}_{\widetilde{t}} \times \{y^{n}\}.
                    \end{equation}
                    This is a disjoint union, and therefore for $y^{n} \in T^{Y^{n}}_{\widetilde{t}}$
                    \begin{equation}
                    \left|T_{\widetilde{t}}^{X^{n},Y^{n}}\right| = \left|T^{Y^{n}}_{\widetilde{t}}\right| \cdot \left|T^{X^{n}| y^{n}}_{\widetilde{t}}\right|.
                    \end{equation}
                \end{claim}
                A complete proof of Claim~\ref{clm:cond_typ_class} is given in the \hyperref[proof:cond_typ_class]{Appendix}.
                \begin{lemma} \label{prop:cond_typ_class_card}
                    The cardinality of $T_{\widetilde{t}}^{X^{n}|y^{n}}$ is bound by 
                    \begin{equation}
                        2^{nH\left(X| Y\right)_{\widetilde{t}}-|\mathcal{X}|\cdot |\mathcal{Y}|\log\left(n+1\right)} \leq \left|T_{\widetilde{t}}^{X^{n}|y^{n}}\right| \leq 2^{nH\left(X| Y\right)_{\widetilde{t}}+|\mathcal{X}|\cdot |\mathcal{Y}|\log\left(n+1\right)}.
                    \end{equation}
                \end{lemma}
             \begin{proof}[Proof of Lemma~\ref{prop:cond_typ_class_card}] \label{proof:cond_typ_class_card}
                The cardinality of each type class $\left|T_{\widetilde{t}}^{X^{n},Y^{n}}\right|$ and $\left|T^{Y^{n}}_{\widetilde{t}}\right|$ is bound by Proposition~\ref{prop:card_typ_class}. The result follows by applying Claim \ref{clm:cond_typ_class}.
            \end{proof}
            \subsubsection{$\delta$-typical sets}
            We will now introduce $\delta$-typical sets for a fixed joint type $\widetilde{t} \in \mathcal{T}^{\mathcal{X}^{n},\mathcal{Y}^{n}}$ and $\delta \geq 0$.
            \begin{definition}
            The $\delta$-typical set $T^{X^{n}\backslash Y^{n}}_{ \widetilde{t}, \delta}$ is defined as
                \begin{equation}
                    T^{X^{n}\backslash Y^{n}}_{ \widetilde{t}, \delta} = \bigcup_{\substack{t \in \mathcal{T}^{\mathcal{X}^{n}, \mathcal{Y}^{n}}\\ \|t-\widetilde{t}\|_{1} \leq \delta}}T^{X^{n} \backslash Y^{n}}_{t}.
                \end{equation}
            \end{definition}
            Note that in this case, $T^{X^{n}\backslash Y^{n}}_{ \widetilde{t}, \delta}$ and $T^{X^{n}}_{ \widetilde{t}, \delta}$ are not necessarily equivalent.
    \begin{definition}
    The conditionally $\delta$-typical set $T^{X^{n}|y^{n}}_{ \widetilde{t}, \delta}$ is defined as
                \begin{equation}
                    T^{X^{n}|y^{n}}_{ \widetilde{t}, \delta} = \bigcup_{\substack{t \in \mathcal{T}^{\mathcal{X}^{n}, \mathcal{Y}^{n}}\\ \|t-\widetilde{t}\|_{1} \leq \delta}}T^{X^{n} | y^{n}}_{t}.
                \end{equation}
            \end{definition}
            \begin{proposition} The cardinality of the conditionally $\delta$-typical set $T^{X^{n}|y^{n}}_{ \widetilde{t}, \delta}$ can be bounded as follows
            \label{prop:card_cond_typ_set}
                            \begin{equation}
                               2^{n(H\left(X| Y\right)_{\widetilde{t}}-\gamma\left( \left|\mathcal{X}\right|, \delta \right))-|\mathcal{X}| \cdot |\mathcal{Y}|\log\left(n+1\right)} \leq \left| T^{X^{n}|y^{n}}_{ \widetilde{t}, \delta} \right| \leq 2^{n\left(H(X|Y)_{\widetilde{t}}+\gamma\left( \left|\mathcal{X}\right|, \delta \right)\right)+2\left|\mathcal{X}\right|\cdot|\mathcal{Y}|\log\left(n+1\right)},
                            \end{equation}
                            where $\gamma\left(\left|\mathcal{X}\right|, \delta \right) = \delta\log_{2}\left(\left|\mathcal{X}\right|\right)+h_{2}\left(\delta\right)$.
                        \end{proposition}
                    A complete proof of Proposition~\ref{prop:card_cond_typ_set} is given in the \hyperref[proof:card_cond_typ_set]{Appendix}.

            With these definitions, we can introduce a distribution $p_{M|X}$ on the alphabet $\mathcal{M}\times \mathcal{X}$, and the typical sets of this distribution and the fixed type $\widetilde{t} \in \mathcal{T}^{\mathcal{X}^{n}, \mathcal{Y}^{n}}$ corresponding to parameters $\delta, \delta' \geq 0$. For simplicity we will drop the subscript indicating that $p = p_{M|X}$. Note that 
            \begin{equation}
            \label{eqn:cond_dist}
                p^{n}(m^{n}|x^{n}) = p^{n}(m^{n}|x^{n}, y^{n})= p(m_{1}|x_{1}, y_{1})p(m_{2}|x_{2}, y_{2})... p(m_{n}|x_{n}, y_{n}).
            \end{equation}
             For $\widetilde{t} \in \mathcal{T}^{\mathcal{X}^{n}, \mathcal{Y}^{n}}$, and a pair of sequences $(x^{n}, y^{n}) \in T^{X^{n}, Y^{n}}_{\widetilde{t}}$, we can define the conditionally typical set $T^{M^{n}|x^{n}, y^{n}}_{\widetilde{t},(p, \delta')}$ as
                \begin{equation}
                    T^{M^{n}|x^{n}, y^{n}}_{\widetilde{t},(p, \delta')}= \bigcup_{\substack{t \in \mathcal{T}^{\mathcal{M}^{n}, \mathcal{X}^{n}, \mathcal{Y}^{n}}\\ \|t -p\cdot\widetilde{t}\|_{1} \leq \delta'}}T^{M^{n} |x^{n}, y^{n}}_{t}.
                \end{equation}
                Note that in this case,
                \begin{equation}
                    T^{M^{n}|x^{n}, y^{n}}_{\widetilde{t},(p, \delta')}= T^{M^{n}|x^{n}}_{\widetilde{t},(p, \delta')}.
                \end{equation}
            \begin{proposition}
    \label{prop:unit_prob}
    Let $(x^{n}, y^{n}) \in T^{X^{n}, Y^{n}}_{\widetilde{t}}$, then 
    \begin{equation}
       p^{n}\left(T^{M^{n}|x^{n}, y^{n}}_{\widetilde{t}, (p,\delta')}\right) \geq 1-2^{-\frac{n\delta'^{2}}{2\ln 2}+2|\mathcal{M}|\cdot|\mathcal{X}|\cdot|\mathcal{Y}|\log(n+1)}.
    \end{equation}
\end{proposition}
A complete proof of Proposition~\ref{prop:unit_prob} is given in the \hyperref[proof:unit_prob]{Appendix}.

            We can also define the jointly typical set $T^{M^{n}, X^{n}, Y^{n}}_{\widetilde{t},  (p, \delta')}$ as
                \begin{equation}
                    T^{M^{n}, X^{n}, Y^{n}}_{\widetilde{t},  (p, \delta')}= \bigcup_{(x^{n}, y^{n})\in T^{X^{n}, Y^{n}}_{\widetilde{t}}}T^{M^{n}|x^{n}, y^{n}}_{\widetilde{t},(p, \delta')}\times \{x^{n}\}\times \{y^{n}\},
                \end{equation}
              and $T^{M^{n}, X^{n}, Y^{n}}_{(\widetilde{t}, \delta), (p, \delta')}$ as
                \begin{equation}
                    T^{M^{n}, X^{n}, Y^{n}}_{(\widetilde{t}, \delta), (p, \delta')}= \bigcup_{\substack{t \in \mathcal{T}^{\mathcal{X}^{n}, \mathcal{Y}^{n}}\\ \|t -\widetilde{t}\|_{1} \leq \delta}}T^{M^{n}, X^{n}, Y^{n}}_{t,  (p, \delta')}.
                \end{equation}
            We can define the set $T^{M^{n}|x^{n}\backslash Y^{n}}_{(\widetilde{t}, \delta), (p,  \delta')}$ as
                \begin{equation}
                   T^{M^{n}|x^{n}\backslash Y^{n}}_{(\widetilde{t}, \delta), (p,  \delta')}= \left \{m^{n} \in \mathcal{M}^{n}:\exists y^{n} \mbox{ s.t. } (m^{n}, x^{n}, y^{n}) \in T^{M^{n}, X^{n}, Y^{n}}_{(\widetilde{t}, \delta), (p, \delta')}\right\}.
                \end{equation}
\begin{proposition}
\label{prop:card_cond_bs_set}
For a fixed joint type $\widetilde{t} \in \mathcal{T}^{\mathcal{X}^{n}, \mathcal{Y}^{n}}$, and $y^{n} \in T^{Y^{n}}_{\widetilde{t}}$, an upper bound on the cardinality of $T^{M^{n}|y^{n}\backslash X^{n}}_{(\widetilde{t}, \delta), (p,  \delta')}$ is given by
    \begin{equation}
        \left|T^{M^{n}|y^{n}\backslash X^{n}}_{(\widetilde{t}, \delta), (p,  \delta')}\right| \leq 2^{n(H\left(M| Y\right)_{\widetilde{t}\cdot p}+\gamma(|\mathcal{M}|, \left|\mathcal{X}\right|\cdot\delta'))+2|\mathcal{M}|\cdot |\mathcal{Y}|\log\left(n+1\right)},
    \end{equation}
    where the conditional entropy is that of the marginal distribution $\sum_{x\in\mathcal{X}}p(m|x)\cdot \widetilde{t}(x, y)$.
\end{proposition}
 A complete proof of Proposition~\ref{prop:card_cond_bs_set} is given in the \hyperref[proof:card_cond_bs_set]{Appendix}.
\begin{proposition} 
\label{prop:merge_set}
The set $T^{M^{n}, X^{n}, Y^{n}}_{(\widetilde{t}, \delta), (p, \delta')}$ is contained in the jointly typical set $T^{M^{n}, X^{n}, Y^{n}}_{(p \cdot \widetilde{t}, \delta'+\delta)}$, 
\begin{equation}
    T^{M^{n}, X^{n}, Y^{n}}_{(\widetilde{t}, \delta), (p, \delta')} \subseteq  T^{M^{n}, X^{n}, Y^{n}}_{(p \cdot \widetilde{t}, \delta'+\delta)}. 
\end{equation}
\end{proposition}
A complete proof of Proposition~\ref{prop:merge_set} is given in the \hyperref[proof:merge_set]{Appendix}.

    \section{Prior-free compression of one-way communication protocols} \label{sec:single}
        \subsection{Prior-Free Slepian Wolf Coding} \label{sec:sw}

            We will first consider the simple setting in which Alice and Bob are each given an $n$ symbol sequence $x^{n} \in \mathcal{X}^{n}$ and $y^{n} \in \mathcal{Y}^{n}$ respectively. The goal is for Bob to obtain a copy of Alice's sequence $x^{n}$. We will show how Alice and Bob can accomplish this using a noiseless bit channel and without any pre-shared randomness.

            We will consider three variants of this scenario, in order to iteratively exhibit the various tools developed in the previous section. First, we will assume that Alice and Bob are provided, as an additional shared input, an estimate of the joint type of their inputs. Subsequently, we will remove the need for this additional information, at the expense of Bob having to communicate a small message to Alice. Finally, we will show that if Bob's input $y^{n}$ is generated from Alice's input through the use of as noisy channel, then the task can be completed without both the knowledge of the joint type, and the additional message from Bob to Alice.

		\subsubsection{$1$-round protocol with an estimate of the joint type given as an additional input}

            \begin{theorem}\label{thm:sw_1}
                Let $\mathcal{X}$ and $\mathcal{Y}$ be finite alphabets, and fix an $n \in \mathbb{N}$. There exists a $1$-round protocol taking $(x^{n}, y^{n}) \in \mathcal{X}^{n}\times \mathcal{Y}^{n}$ as Alice and Bob's respective inputs, as well as a parameter $\delta > 0$ and a joint type ${\widetilde{t} \in \mathcal{T}^{\mathcal{X}^{n}, \mathcal{Y}^{n}}}$ as shared inputs, with the guarantee that $(x^{n}, y^{n}) \in T^{X^{n}, Y^{n}}_{\widetilde{t}, \delta}$, that satisfies the following conditions:
                \begin{itemize}
                    \item The communication cost is at most $n\left(H(X|Y)_{t}+\eta_{1}(n, \delta)\right)$;
                    \item No shared randomness is required;
                    \item There is an event $E_{good}$ such that:
                    \begin{itemize}
                        \item $\Pr[\lnot E_{good}] \leq 2^{-n\delta}$;
                        \item Conditional on $E_{good}$, Bob can produce a copy $\hat{x}^{n}$ of $x^{n}$ from his input and the protocol transcript;
                    \end{itemize}
                \end{itemize}
                where $\eta_{1}(n, \delta)=2\gamma\left( \left|\mathcal{X}\right|, \delta \right)+\frac{2}{n}\left|\mathcal{X}\right|\cdot|\mathcal{Y}|\log\left(n+1\right)+\frac{1}{n}\log_{2}(n)+3\delta+O\left(\frac{1}{n}\right).$
            \end{theorem}
            \begin{proof}
                We will first show how to accomplish such a task if Alice and Bob can use an unlimited amount of shared randomness. Consider the following protocol:
                \begin{enumerate}
                    \item Alice and Bob sample the shared randomness to arrange the elements of $\mathcal{X}^{n}$ in an ordered list $\mathcal{L}$ in a uniformly random fashion.
                    \item Alice sends Bob the first $nC$ bits of the binary expansion of the position of her sequence in the list $\mathcal{L}$, choosing $C= H(X|Y)_{\widetilde{t}}+\gamma\left( \left|\mathcal{X}\right|, \delta \right)+\frac{2}{n}\left|\mathcal{X}\right|\cdot|\mathcal{Y}|\log\left(n+1\right)+\delta+\frac{1}{n}$.
                    \item Bob removes any sequences in $\mathcal{L}$ that are not in positions that correspond to the $nC$ bits sent by Alice. Denote this new list $\mathcal{L}_{c}$. He also determines the conditionally typical set $T^{X^{n}|y^{n}}_{ \widetilde{t},\delta}$ of $\widetilde{t}$.
                    \item Bob looks in the intersection $\mathcal{L}_{c} \cap T^{X^{n}|y^{n}}_{ \widetilde{t},\delta}$ for a sequence $\hat{x}^{n}$, and declares an error if $\left|\mathcal{L}_{c} \cap T^{X^{n}|y^{n}}_{ \widetilde{t},\delta} \right| \neq 1$. Otherwise, he outputs $\hat{x}^{n}$.
                \end{enumerate}
                An error can occur in the protocol if there is another sequence $\widetilde{x^{n}}$ in the list $\mathcal{L}_{c}$ that is also in the set $T^{X^{n}|y^{n}}_{\widetilde{t}, \delta}$, in which case Bob will not be able to determine whether $x^{n}$ or $\widetilde{x^{n}}$ is Alice's input.
                To bound the probability of this error occurring, we use the fact that the positions in the list $\mathcal{L}$ are assigned uniformly at random,
                \begin{align}
                    &\Pr\left[\exists \widetilde{x^{n}} \in T^{X^{n}|y^{n}}_{ \widetilde{t},\delta} : \widetilde{x^{n}} \neq x^{n} : \mbox{ and } \widetilde{x^{n}} \in \mathcal{L}_{c}   \right]
                    \\&\leq \sum_{\widetilde{x^{n}} \in T^{X^{n}|y^{n}}_{ \overline{t}, \delta}}\Pr\left[\widetilde{x^{n}} \neq x^{n} \mbox{ and } \widetilde{x^{n}} \in \mathcal{L}_{c}\right] \\
                    &\leq \sum_{\substack{\widetilde{x^{n}} \in T^{X^{n}|y^{n}}_{ \overline{t}, \delta}\\
                    \widetilde{x^{n}} \neq x^{n}}}\Pr\left[\widetilde{x^{n}} \in \mathcal{L}_{c}\right] 
                    \\&\leq \left|T^{X^{n}|y^{n}}_{ \overline{t}, \delta}\right|\frac{1}{2^{nC}} 
                    \\&\leq 2^{n\left(H(X|Y)_{\widetilde{t}}+\gamma\left( \left|\mathcal{X}\right|, \delta \right)\right)+2\left|\mathcal{X}\right|\cdot|\mathcal{Y}|\log\left(n+1\right)}\frac{1}{2^{nC}}.
                \end{align}
                Alice and Bob can therefore take $C = H(X|Y)_{\widetilde{t}}+\gamma\left( \left|\mathcal{X}\right|, \delta \right)+\frac{2}{n}\left|\mathcal{X}\right|\cdot|\mathcal{Y}|\log\left(n+1\right)+\delta+\frac{1}{n},$
                in which case 
                \begin{equation}
                    \Pr\left[\exists \widetilde{x^{n}} \in T^{X^{n}|y^{n}}_{ \widetilde{t},\delta} : \widetilde{x^{n}} \neq x^{n} : \mbox{ and } \widetilde{x^{n}} \in \mathcal{L}_{c}   \right] \leq \frac{2^{-n\delta}}{2}.
                \end{equation}

                To avoid using shared randomness, we will first show how similar guarantees can be obtained by using $O(n\delta)$ bits of shared randomness. The result then follows from allowing Alice to privately sample the randomness, and then communicate it to Bob. The following claim is obtained from a standard Chernoff bound argument.
                \begin{claim}
                \label{clm:derandom_sw}
                    There exists a protocol that achieves the communication cost stated in Theorem~\ref{thm:sw_1}, that uses only $\log n+2n\delta +O(1)$ bits of shared randomness. There is also an event $E_{good}$, with $\Pr[E_{good}] \geq 1 - 2^{-n\delta}$, such that conditioned on $E_{good}$, Bob can produce $\hat{x}^{n}$ from his input and the protocol transcript.
                \end{claim}
                \begin{proof}
                    We follow the proof of Newman's theorem \cite{newman_1991} given in \cite{Rao_Yehudayoff_2020}, which uses the probabilistic method to find the desired random strings. Let us pick $s$ independent random strings to use for the unbounded randomness protocol. Let $\delta'>0$, then for any input $(x^{n}, y^{n})$, the probability that a $\frac{2^{-n\delta}}{2}+\delta'$ fraction of the $s$ strings lead to $\hat{x}^{n} \neq x^{n}$ is at most $2^{-\Omega(s\delta'^{2})}$ by the Chernoff bound. Letting
                    \begin{equation}
                        s = O\left(\frac{n\log\left(|\mathcal{X}|\cdot|\mathcal{Y}| \right)}{\delta'^{2}} \right),
                    \end{equation}
                    we find that this probability is smaller than $\left(|\mathcal{X}|^{n}\cdot|\mathcal{Y}|^{n}\right)^{-1}$, so by the union bound, the probability that more than a $\frac{2^{-n\delta}}{2}+\delta'$ fraction of the $s$ strings gives the wrong answer for any input $(x^{n}, y^{n}) \in \mathcal{X}^{n} \times \mathcal{Y}^{n}$ is less than $1$. 

                    Our bounded randomness protocol is then as follows. Sample one of these $s$ strings at random, and then run the corresponding unbounded randomness protocols. Let $E_{good}$ be the event that the sampled string leads to $\hat{x}^{n} = x^{n}$ on input $(x^{n}, y^{n})$, then 
                    \begin{equation}
                        \Pr[E_{good}] \geq 1-\frac{2^{-n\delta}}{2}-\delta'.
                    \end{equation}
                    Picking $\delta' = \frac{2^{-n\delta}}{4}$ leads to $\Pr[E_{good}] \geq 1 - 2^{-n\delta}$, and the number of bits of shared randomness is then 
                    \begin{equation}
                        \log(s) = \log(n) +\log\log\left(|\mathcal{X}|\cdot|\mathcal{Y}|\right)+2n\delta+O(1)
                    \end{equation}
                     as desired.
                \end{proof}
                With the result of Claim~\ref{clm:derandom_sw}, and applying the Alicki-Fannes inequality to define the entropy in terms of the true joint type $t_{x^{n}, y^{n}}$, the communication cost of the protocol requiring no pre-shared randomness is bound by $n\left(H(X|Y)_{t}+\eta_{1}(n, \delta)\right)$, with
                \begin{equation}
                    \eta_{1}(n, \delta)=2\gamma\left( \left|\mathcal{X}\right|, \delta \right)+\frac{2}{n}\left|\mathcal{X}\right|\cdot|\mathcal{Y}|\log\left(n+1\right)+\frac{1}{n}\log(n)+3\delta+O\left(\frac{1}{n}\right).
                \end{equation}
\end{proof}
		       \subsubsection{$2$-round protocol without an estimate of the joint type given}

            By combining Theorem~\ref{thm:sw_1} with the sampling protocol of Lemma~\ref{lem:estimating} to obtain an estimate $\widetilde{t}$ of the joint type of Alice and Bob's inputs $(x^{n}, y^{n})$, we obtain, in the following theorem, a two-round protocol with similar guarantees.
                \begin{theorem}\label{thm:sw_2}
                    Let $\mathcal{X}$ and $\mathcal{Y}$ be finite alphabets, and fix an $n \in \mathbb{N}$. There exists a $2$-round protocol taking $(x^{n}, y^{n}) \in \mathcal{X}^{n}\times \mathcal{Y}^{n}$ as Alice and Bob's respective inputs, as well as parameters $\delta, \delta_{s} > 0$, as shared inputs, that satisfies the following conditions:
                \begin{itemize}
                    \item The first message from Bob to Alice is of fixed size $n\left(\delta_{s}\log_{2}(|\mathcal{Y}|)+ \frac{1}{n}\log(n) +2\delta+O\left(\frac{1}{n}\right)\right)$;
                    \item The second message from Alice to Bob is of variable length, but always of size at most $ n(\delta_{s}+1)\log_{2}\left(|\mathcal{X}| \right)$;
                    \item No shared randomness is required;
                    \item There is an event $E_{good}$ such that $\Pr[\lnot E_{good}] \leq 2^{-n\delta^{2}}$, and conditional on $E_{good}$:
                    \begin{itemize}
                        \item The communication cost is at most $n\left(H(X|Y)_{t}+\delta_{s}\log_{2}(|\mathcal{X}|\cdot|\mathcal{Y}|)+\eta_{1}(n, \delta)\right)$;
                        \item Bob can produce a copy $\hat{x}^{n}$ of $x^{n}$ from his input and the protocol transcript;
                    \end{itemize}
                \end{itemize}
                where $\eta_{1}$ is as defined in Theorem~\ref{thm:sw_1}.
                \end{theorem}
                \begin{proof}
                    We first show how to accomplish such a task if Alice and Bob can use an unbounded quantity of shared randomness. Consider the following protocol
                    \begin{enumerate}
                        \item Alice and Bob perform the  (unbounded randomness) protocol of Lemma~\ref{lem:estimating} to obtain an estimate $\widetilde{t}$ with $m = n \cdot \delta_{s}$.
                        \item Alice and Bob perform the protocol of Theorem~\ref{thm:sw_1} using $\widetilde{t}$.
                    \end{enumerate}
                    By Lemma~\ref{lem:estimating}, the probability that $(x^{n}, y^{n}) \not\in T^{X^{n}, Y^{n}}_{\overline{t}, \delta}$ is upper bounded by $\frac{2^{-n\delta'}}{4}$, where
                    \begin{equation}
                        \delta' = 2\delta_{s}\delta^{2}\log e+\frac{1}{n}\log\left|\mathcal{X}\right|\cdot \left|\mathcal{Y}\right|+\frac{3}{n},
                    \end{equation}
                    and the probability that the protocol of Theorem~\ref{thm:sw_1} will fail is bounded by $\frac{2^{-n\delta}}{2}$. By the union bound, the probability that either event occurs is at most $\frac{2^{-n\delta}}{2}+\frac{2^{-n\delta'}}{4} \leq \frac{3\cdot2^{-n\delta'}}{4}$, where we use the fact that $\delta' < \delta$.

                    A similar derandomization argument as in the proof of Claim~\ref{clm:derandom_sw} gives us the following claim. The details of the proof are provided in the \hyperref[proof:derandom_sw_2]{Appendix}.
                    \begin{claim}
                        \label{clm:derandom_sw_2}
                        There is a protocol that achieves the communication cost stated above, but using only $\log(n)+\log\log\left(|\mathcal{X}|\cdot|\mathcal{Y}| \right)+2n\delta'+O(1)$ bits of shared randomness, and on event $E_{good}$, with $\Pr[E_{good}]\geq 1 - 2^{-n\delta'}$, Bob can produce $\hat{x}^{n} = x^{n}$ from his input and the protocol transcript.
                    \end{claim}
                    The desired protocol follows by letting Bob privately sample the randomness and communicate it to Alice. The communication cost is then at most $n\left(H(X|Y)_{t}+\delta_{s}\log_{2}(|\mathcal{X}|\cdot|\mathcal{Y}|)+\eta_{1}(n, \delta)\right)$.
                \end{proof}

 		      \subsubsection{$1$-round protocol with a noisy estimate of the joint type given to the receiver}

                In the case where Bob's input $y^{n}$ is obtained from Alice's input $x^{n}$ through the use of some noisy channel $y=N(x)$ associated with the output distribution $p_{Y|X}(y|x)$, i.e. $$p_{Y|X}^{n}(y^{n}|x^{n}) = p_{Y|X}(y_{1}|x_{1})\cdot p_{Y|X}(y_{2}|x_{2})\cdot ...\cdot p_{Y|X}(y_{n}|x_{n}),$$ we can avoid the need for back communication from Bob in order to estimate the joint type. Moreover, this allows us to exhibit the additional machinery we have developed for the reverse Shannon theorem in a simpler message transmission setting. We obtain the following.

                \begin{theorem}
                \label{thm:sw_3}
                    Fix finite alphabets $\mathcal{X}$, $\mathcal{Y}$, and $n \in \mathbb{N}$. There is a $1$-round protocol taking $(x^{n}, y^{n}) \in \mathcal{X}^{n} \times \mathcal{Y}^{n}$ as Alice's and Bob's respective inputs, for which $y^{n}$ is obtained from $x^{n}$ as the output of $n$ independent instances of some memoryless channel $N(x)$, with output distribution $p_{Y|X}(y|x)$, as well as a parameter $\delta>0$ and a description of $p_{Y|X}(y|x)$ as extra shared inputs, that satisfies the following:

                    Define $\delta' = \frac{\delta^{2}}{2\ln 2}-\frac{2}{n}|\mathcal{Y}|\cdot|\mathcal{X}|\log(n+1)-\frac{2}{n}$. 
                    
                    \begin{itemize}
                        \item The communication cost is $n\left(H(X|Y)_{t\cdot p}+\frac{\left|\mathcal{X}\right|}{n}\log_{2}(n)+\eta_{1}(n, \delta)\right)$, where $t\cdot p$ denotes the distribution $t_{x^{n}}\cdot p_{Y|X}$;
                        \item No shared randomness is required;
                        \item There is an event $E_{good}$ such that:
                        \begin{itemize}
                            \item $\Pr[\lnot E_{good}] \leq 2^{-n\delta'}$;
                            \item Conditional on $E_{good}$, Bob can produce a copy $\hat{x}^{n}$ of $x^{n}$ from his input and the protocol transcript;
                        \end{itemize}
                    \end{itemize}
                    where $\eta_{1}$ is defined in Theorem~\ref{thm:sw_1}.
                \end{theorem}
                \begin{proof}
                    We will first show how to accomplish such a task if Alice and Bob can use an unbounded amount of pre-shared randomness.
                    Consider the following protocol:
                    \begin{enumerate}
                        \item Alice transmits $t_{x^{n}}$ to Bob using $|\mathcal{X}|\log_{2}(n)$ bits of communication.
                        \item Both players compute the set $T^{X^{n}, Y^{n}}_{t, (p, \delta)}$.
                        \item They perform the (unbounded randomness) protocol from the proof of Theorem~\ref{thm:sw_1}, replacing $T^{X^{n}, Y^{n}}_{\widetilde{t}, \delta}$ with $T^{X^{n}, Y^{n}}_{t, (p, \delta)}$.
                    \end{enumerate}
                   By Proposition~\ref{prop:unit_prob},
                   \begin{align}
                       \Pr\left[(x^{n}, y^{n}) \not\in T^{X^{n}, Y^{n}}_{t, (p, \delta)}\right]&=\Pr\left[y^{n} \not\in T^{Y^{n}|x^{n}}_{t, (p, \delta)}\right]\\ 
                       &\leq 2^{-\frac{n\delta^{2}}{2\ln 2}+2|\mathcal{Y}|\cdot|\mathcal{X}|\log(n+1)}\\
                       &=\frac{2^{-n\delta'}}{4},\\
                       \intertext{with} \delta' &= \frac{\delta^{2}}{2\ln 2}-\frac{2}{n}|\mathcal{Y}|\cdot|\mathcal{X}|\log(n+1)-\frac{2}{n}.
                   \end{align}
                   
                   If $(x^{n}, y^{n}) \in T^{X^{n}, Y^{n}}_{t, (p, \delta)}$, then by Theorem~\ref{thm:sw_1}, the probability that $\hat{x}^{n}\neq x^{n}$ is at most $\frac{2^{-n\delta}}{2}$, and total probability of error is at most $\frac{2^{-n\delta}}{2}+\frac{2^{-n\delta'}}{4} \leq \frac{3\cdot2^{-n\delta'}}{4}$, where we use the fact that $\delta' < \delta$. An approach similar to that of Claim~\ref{clm:derandom_sw_2} allows us to attain the guarantees stated in the theorem.
                \end{proof}
        \subsection{Prior-Free Reverse Shannon Theorem with Side Information}\label{sec:rst}

            We will now consider the setting where Alice and Bob both have a description of some noisy channel $N$ associated with the output distribution $p_{M|X}$, where an output $m=N(x)$ is distributed according to the (conditional) random variable $M\sim p_{M|X}$. This distribution and $t_{x^{n}, y^{n}}$ will often be abbreviated to $p$ and $t$ respectively when used as subscripts. Alice and Bob are each given an $n$ symbol sequence $x^{n} \in \mathcal{X}^{n}$ and $y^{n}\in \mathcal{Y}^{n}$ respectively. They would like to simulate $n$ independent uses of the channel on Alice's input, and each obtain a copy of the $n$ simulated outputs $\hat{m}^{n}$. They have access to a noiseless bit channel and shared random bits.

            We will consider two variants of this scenario. First, we will assume that Alice and Bob are provided, as an additional shared input, an estimate $\widetilde{t}$ of the joint type of their inputs. Subsequently, we will remove the need for this additional information, at the expense of Bob having to communicate a small message to Alice.
            \subsubsection{$1$-round protocol with an estimate of the joint type given as an additional input}

            \begin{theorem} \label{thm:rst_1}
                Let $\mathcal{M}, \mathcal{X},$ and $\mathcal{Y}$ be finite alphabets, and fix an $n \in \mathbb{N}$. There exists a $1$-round protocol taking $(x^{n}, y^{n}) \in \mathcal{X}^{n}\times \mathcal{Y}^{n}$ as Alice and Bob's respective inputs, as well as parameters $\delta, \delta', \delta'' > 0$, a description of the noisy channel $N$ associated with the output distribution $p_{M|X}$, and a joint type $\widetilde{t} \in \mathcal{T}^{\mathcal{X}^{n}, \mathcal{Y}^{n}}$, with the guarantee that $(x^{n}, y^{n}) \in T^{X^{n}, Y^{n}}_{\widetilde{t}, \delta}$ as additional shared inputs, that satisfies the following conditions:

                Define $\delta''' = \left(\frac{\delta'^{2}}{2\ln 2}-\frac{2}{n}|\mathcal{M}|\cdot|\mathcal{X}|\cdot|\mathcal{Y}|\log(n+1)\right)$, $\delta_{\textnormal{min}} = \min\left(\delta, \delta''', \delta'' \right)$ and $\delta_{\textnormal{max}} = \max\left(\delta, \delta', \delta'' \right)$.
                \begin{itemize}
                    \item The communication cost is at most: $n\left(I\left(M;X|Y\right)_{t\cdot p}+\eta_{2}(n, \delta_{\textnormal{max}})\right)$;
                    \item The amount of shared randomness needed is at most $n\left(H\left(M|X, Y\right)_{t\cdot p}+\frac{1}{n}\log\log(e)\right)$;
                    \item There is an event $E_{good}$ such that $\Pr[\lnot E_{good}] \leq 2^{-n \delta_{\textnormal{min}}^{2}+3}$, and conditional on $E_{good}$:
                    \begin{itemize}
                    \item Alice and Bob each produce a copy of the same output $\hat{m}^{n}$;
                        \item It holds that for every $(x^{n}, y^{n}) \in \mathcal{X}^{n} \times \mathcal{Y}^{n}$, if $\widetilde{p^{n}}(m^{n}|x^{n}, y^{n})$ is the probability that Alice and Bob obtain the output $m^{n}$ from the protocol, then
                        \begin{equation}
                           \left\|\widetilde{p^{n}}(m^{n}|x^{n}, y^{n})-p^{n}(m^{n}|x^{n}, y^{n})\right\|_{1} \leq 2^{-n\delta'''} + \delta'';
                        \end{equation}
                    \end{itemize}
                \end{itemize}
                where
                $\eta_{2}(n, \delta_{\textnormal{max}})=5\gamma(|\mathcal{M}|, \left|\mathcal{X}\right|\cdot\left|\mathcal{Y}\right|\cdot\delta_{max})+\frac{4}{n}|\mathcal{M}| \cdot|\mathcal{X}|\cdot |\mathcal{Y}|\log\left(n+1\right)+\frac{2}{n}\log(n)+3\delta_{\textnormal{max}}+O\left(\frac{1}{n}\right)$.
             \end{theorem}
             \begin{proof}
                We will first show how to accomplish such a task if Alice and Bob can use an unlimited amount of shared randomness. Consider the following protocol:
                \begin{enumerate}
                    \item Alice applies the channel $N^{n}$ on her input $x^{n}$ to obtain the message $m^{n}$.
                    \item Alice determines the joint type $t\left(m^{n}, x^{n}\right)$ of the message and her input. She verifies if the message is in the typical set $T^{M^{n}|x^{n}\backslash Y^{n}}_{(\widetilde{t}, \delta), (p,  \delta')}$,
                    and if so, she sends the type of the message $t\left(m^{n}\right)$ to Bob using $|\mathcal{M}|\cdot \log_{2}\left(n+1\right)$ bits, otherwise she declares an error.
                    \item Alice and Bob remove any messages from the list $\mathcal{L}$ that do not have type $t(m^{n})$.
                    \item Alice and Bob sample the shared randomness to arrange the remaining elements of $\mathcal{M}^{n}$ in a uniformly random order. Let $\mathcal{L}_{t}$ denote this new list.
                    \item  Alice and Bob sample the shared randomness to obtain $nR$ shared random bits, choosing
                    \begin{equation}
                        R = H\left(M|X, Y\right)_{\widetilde{t}\cdot p}-\gamma(\left|\mathcal{M}\right|, \delta)-\gamma(\left|\mathcal{M}\right|, \delta')-\frac{1}{n}|\mathcal{M}| \cdot |\mathcal{X}|\log\left(n+1\right)+\frac{1}{n}\log\log(e)-\frac{1}{n}-\frac{1}{n}\log(n).
                    \end{equation}
                    \item They remove any messages from $\mathcal{L}_{t}$ that do not have these bits as the first $nR$ bits in the binary expansion of their position in $\mathcal{L}_{t}$, while maintaining the order of the remaining messages. Let $\mathcal{L}_{r}$ denote this new list.
                    \item Alice randomly chooses a message $m^{n}_{r}$ from $\mathcal{L}_{r}$ that has the correct joint type of $t\left(m^{n}, x^{n}\right)$ with her sequence $x^{n}$. If there are no messages with the correct joint type, she declares an error.
                    \item Alice transmits the first $nC$ bits of the binary expansion of the position of $m^{n}_{r}$ in $\mathcal{L}_{r}$ to Bob, choosing
                    \begin{align}
                        C =I\left(M;X|Y\right)_{\widetilde{t}\cdot p}+\gamma(|\mathcal{M}|, \left|\mathcal{X}\right|\cdot\delta')+\gamma(\left|\mathcal{M}\right|, \delta)+\gamma(\left|\mathcal{M}\right|, \delta')
                        +\frac{3}{n}|\mathcal{M}| \cdot|\mathcal{X}|\cdot |\mathcal{Y}|\log\left(n+1\right)\\+\delta+\frac{1}{n}+\frac{1}{n}\log(n).
                    \end{align}
                    \item Bob removes any messages whose first $nC$ bits in the binary expansion of their position in $\mathcal{L}_{r}$ do not match the $nC$ bits sent by Alice. Let $\mathcal{L}_{r, c}$ denote this new list. He also determines the conditionally typical set $T^{M^{n}|y^{n}\backslash X^{n}}_{(\widetilde{t}, \delta), (p,  \delta')}$ for a sequence $\hat{m}^{n}$ and declares an error if $\left|\mathcal{L}_{r, c}   \cap T^{M^{n}|y^{n}\backslash X^{n}}_{(\widetilde{t}, \delta), (p,  \delta')}\right| \neq 1$. Otherwise, he outputs $\hat{m}^{n}$.
                \end{enumerate}

                An error will occur in the protocol if the message $m^{n}$, generated by Alice through her local use of the channel, is not in $T^{M^{n}|x^{n}\backslash Y^{n}}_{(\widetilde{t}, \delta), (p,  \delta')}$ $(E_{1})$, if there is no message in the subset $\mathcal{L}_{r}$ that has the correct joint type with Alice's input $(E_{2})$, in which case Alice will have no messages to choose from, if the message $m^{n}_{r}$ chosen by Alice from the subset $\mathcal{L}_{r}$ is not in $T^{M^{n}|y^{n}\backslash X^{n}}_{(\widetilde{t}, \delta), (p,  \delta')}$ $(E_{3})$, or if $\left|\mathcal{L}_{r, c} \cap T^{M^{n}|y^{n}\backslash X^{n}}_{(\widetilde{t}, \delta), (p,  \delta')}\right| > 1$ $(E_{4})$, in which case Bob will not be able to determine which message Alice chose. A detailed analysis of the probability of these errors occurring, as well as a proof for the upper bound on the total variation distance, can be found in the \hyperref[proof:rst_1]{Appendix}, where we show that the probability of success of the protocol is at least $1 - 2^{-n\delta_{\textnormal{min}}^{2}+2}$.
                
                To reduce the quantity of shared randomness needed, we will first show how similar guarantees can be obtained by using $n(H\left(M|X, Y\right)_{t\cdot p}+\frac{1}{n}\log\log(e))+O(n\cdot\delta_{\textnormal{min}}^{2})$ bits of shared randomness. The result then follows from allowing Alice to privately sample $O(n\cdot\delta_{\textnormal{min}}^{2})$ bits of randomness, and then communicate it to Bob. The following claim is obtained from a standard Chernoff bound argument, and the proof can be found in the \hyperref[proof:derandom_rst]{Appendix}.
                \begin{claim}
                \label{clm:derandom_rst}
                    There exists a protocol that achieves the communication cost stated in Theorem~\ref{thm:rst_1}, that uses only $n(H\left(M|X, Y\right)_{t\cdot p}+\frac{1}{n}\log\log(e))+\log n+2n\cdot\delta_{\textnormal{min}}^{2} +O(1)$ bits of shared randomness. There is also an event $E_{good}$, with $\Pr[E_{good}] \geq 1 - 2^{-n\cdot\delta_{\textnormal{min}}^{2}+3}$, such that conditioned on $E_{good}$, the claims stated in Theorem~\ref{thm:rst_1} hold.
                \end{claim}
                
                The desired protocol follows by letting Alice privately sample $\log n+2n\cdot\delta_{\textnormal{min}}^{2} +O(1)$ bits of randomness and communicate it to Bob. The communication cost of the protocol requiring only $n\left(H\left(M|X, Y\right)_{t\cdot p}+\frac{1}{n}\log\log(e)\right)$ bits of pre-shared randomness is bound by $n\left(I\left(M;X|Y\right)_{t\cdot p}+\eta_{2}(n, \delta_{max})\right)$.
            \end{proof}
            
            \subsubsection{$2$-round protocol without an estimate of the joint type given}
                By combining Theorem~\ref{thm:rst_1} with the sampling protocol of Lemma~\ref{lem:estimating} to obtain a description of $\widetilde{t}$, we obtain, in the following theorem, a two-round protocol with similar guarantees. 
            
            \begin{theorem}\label{thm:rst_2}
                Let $\mathcal{M}, \mathcal{X},$ and $\mathcal{Y}$ be finite alphabets, and fix an $n \in \mathbb{N}$. There exists a $2$-round protocol taking $(x^{n}, y^{n}) \in \mathcal{X}^{n}\times \mathcal{Y}^{n}$ as Alice and Bob's respective inputs, as well as parameters $\delta, \delta', \delta'', \delta_{s} > 0$, and a description of the noisy channel $N$, associated with the output distribution $p_{M|X}$, as additional shared inputs, that satisfies the following conditions:

                Define $\delta''' = \left(\frac{\delta'^{2}}{2\ln 2}-\frac{2}{n}|\mathcal{M}|\cdot|\mathcal{X}|\cdot|\mathcal{Y}|\log(n+1)\right)$, $\delta_{\textnormal{min}} = \min\left(\delta, \delta''', \delta'' \right)$ and $\delta_{\textnormal{max}} = \max\left(\delta, \delta', \delta'' \right)$.
                \begin{itemize}
                \item The first message from Bob to Alice is of fixed size $n\left(\delta_{s}\log_{2}(|\mathcal{Y}|)+\frac{1}{n}\log(n)+2\delta_{\textnormal{min}}^{2}+O\left(\frac{1}{n}\right)\right)$;
                    \item The second message from Alice to Bob is of variable length, but always of size at most $ n\left(\delta_{s}\log_{2}\left(|\mathcal{X}| \right)+\log_{2}\left(|\mathcal{M}| \right)\right)$;
                    \item There is an event $E_{good}$ such that $\Pr[\lnot E_{good}] \leq 2^{-n\cdot\delta_{\textnormal{min}}^{2}+3}$, and conditional on $E_{good}$:
                    \begin{itemize}
                        \item The communication cost is at most $n\left(I\left(M;X|Y\right)_{t\cdot p}+\delta_{s}\log\left(\left|\mathcal{X}\right|\cdot\left|\mathcal{Y}\right|\right)+\eta_{2}(n, \delta_{\textnormal{max}})\right)$;
                        \item The amount of shared randomness needed is at most $n\left(H\left(M|X, Y\right)_{t\cdot p}+\frac{1}{n}\log\log(e)\right)$;
                        \item Alice and Bob each produce a copy of the same output $\hat{m}^{n}$;
                        \item It holds that for every $(x^{n}, y^{n}) \in \mathcal{X}^{n} \times \mathcal{Y}^{n}$, if $\widetilde{p^{n}}(m^{n}|x^{n}, y^{n})$ is the probability that Alice and Bob obtain the output $m^{n}$ from the protocol, then
                        \begin{equation}
                           \left\|\widetilde{p^{n}}(m^{n}|x^{n}, y^{n})-p^{n}(m^{n}|x^{n}, y^{n})\right\|_{1} \leq 2^{-n\delta'''} + \delta'';
                        \end{equation}
                    \end{itemize}
                \end{itemize}
                where $\eta_{2}$ is defined as in Theorem~\ref{thm:rst_1}.
            \end{theorem}

           \begin{proof}
                We first show how to accomplish such a task if Alice and Bob can use an unbounded quantity of shared randomness. Consider the following protocol
                \begin{enumerate}
                    \item Alice and Bob perform step $1$ of Lemma~\ref{lem:estimating} to obtain the strings of symbols to exchange, with $m = n \cdot \delta_{s}$. Denote their strings $s_{A}$ and $s_{B}$. Bob transmits his string $s_{B}$. This concludes the first round of the protocol.
                    \item Alice sends her string $s_{A}$.
                    \item Alice and Bob perform the (unbounded randomness) protocol of Theorem~\ref{thm:rst_1} using $\widetilde{t}$.
                \end{enumerate}
                By Lemma~\ref{lem:estimating}, the probability that $(x^{n}, y^{n}) \not\in T^{X^{n}, Y^{n}}_{\overline{t}, \delta}$ is upper bounded by $2^{-n\delta''''}$, where
                    \begin{equation}
                        \delta'''' = 2\delta_{s}\delta^{2}\log e+\frac{1}{n}\log\left|\mathcal{X}\right|\cdot \left|\mathcal{Y}\right|+\frac{1}{n}.
                    \end{equation}
                    The probability that the protocol of Theorem~\ref{thm:rst_1} will fail is at most $4\cdot2^{-n\cdot\delta_{\textnormal{min}}^{2}}$. By the union bound, the probability that either occurs is at most $5\cdot2^{-n\cdot\delta_{\textnormal{min}}^{2}}$. The players can choose $R$ and $C$ as in the proof of Theorem~\ref{thm:rst_1}.

                    A similar derandomization argument as in the proof of Claim~\ref{clm:derandom_rst} gives us the following. The proof is provided in the \hyperref[proof:derandom_rst_2]{Appendix}.
                    \begin{claim}
                    \label{clm:derandom_rst_2}
                        There exists a protocol that achieves the communication cost stated in Theorem~\ref{thm:rst_2}, that uses only $n(H\left(M|X, Y\right)_{t\cdot p}+\frac{1}{n}\log\log(e))+\log n+2n\delta_{\textnormal{min}}^{2} +O(1)$ bits of shared randomness. There is also an event $E_{good}$, with $\Pr[E_{good}] \geq 1 - 2^{-n\cdot\delta_{\textnormal{min}}^{2}+3}$, such that conditioned on $E_{good}$, the claims stated in Theorem~\ref{thm:rst_2} hold.
                    \end{claim} 
                    
                    The desired protocol follows by letting Bob privately sample $\log n+2n\cdot\delta_{\textnormal{min}}^{2} +O(1)$ bits of randomness and communicate it to Alice. The communication cost of the protocol requiring only $n(H\left(M|X, Y\right)_{t\cdot p}+\frac{1}{n}\log\log(e))$ bits of pre-shared randomness is bound by $n\left(I\left(M;X|Y\right)_{t\cdot p}+\delta_{s}\log\left(\left|\mathcal{X}\right|\cdot\left|\mathcal{Y}\right|\right)+\eta_{2}(n, \delta_{\textnormal{max}})\right)$.
            \end{proof}

    \section{Prior-Free Compression of Interactive Communication Protocols} \label{sec:interactive}
        In this section, Theorem~\ref{thm:rst_2} will be extended to the interactive setting. First, in Theorem~\ref{thm:int_2}, we show that there exists a protocol that simulates $j$ rounds of $n$ channel uses, with an asymptotically decaying probability of error and communication cost bound by the information complexity $IC_{t}$ with respect to the joint type of the inputs $t_{x^{n}, y^{n}}$. The protocol simulates $n$ uses of each channel $N_{i}$, associated with output distribution $p_{M_{i}|X, M_{1},..., M_{j-1}}$ or $p_{M_{i}|Y, M_{1},..., M_{j-1}}$, depending on the parity of $i$, where the channel is applied on the sender's input and the $i-1$ previous messages. We will assume the first sender to be Alice. The transcript of the $i-1$ previous messages $m^{n}_{1},...,m^{n}_{i-1}$ and the associated random variables $M_{1},...,M_{i-1}$ may be abbreviated as $m^{n}_{<i}$ and $M_{<i}$ or $m^{n}_{\leq (i-1)}$ and $M_{\leq (i-1)}$ respectively, and the output distribution may be abbreviated as $p_{i}$ when used as a subscript. In particular, we will use $t_{x^{n}, y^{n}, m^{n}_{<i}}$ to denote the joint type of $x^{n}, y^{n}$, and the $i-1$ previous messages. 
        
We then show in Theorem~\ref{thm:int_3} that, at the expense of slightly increasing the communication cost, there is a $j$ round protocol that simulates $j$ rounds of $n$ channel uses.

\subsection{$(j+1)$-round protocol simulating a $j$-round interactive protocol without an estimate of the joint type given}
        \begin{theorem}\label{thm:int_2}
            Let $\mathcal{X}$, $\mathcal{Y}$ and $\mathcal{M}_{\leq j} = \mathcal{M}_{1}\times...\times\mathcal{M}_{j}$ be finite alphabets, and fix an $n \in \mathbb{N}$. There exists a $j+1$-round protocol taking $\left(x^{n}, y^{n}\right) \in \mathcal{X}^{n} \times \mathcal{Y}^{n}$ as Alice and Bob's respective inputs, as well as parameters $\delta, \delta', \delta'', \delta_{s} > 0$, and a description of $j$ noisy channels $N_{1}, ..., N_{j}$, each associated with an output distribution $p_{M_{i}|X, M_{< i}}$ or $p_{M_{i}|Y, M_{< i}}$ depending on the parity of $i$, as additional shared inputs, that satisfies the following conditions:

            Define $\delta''' = \left(\frac{\delta'^{2}}{2\ln 2}-\frac{2}{n}|\mathcal{M}|\cdot|\mathcal{X}|\cdot|\mathcal{Y}|\log(n+1)\right)$, $\delta_{\textnormal{min}} = \min\left(\delta, \delta''', \delta'' \right)$, $\delta_{\textnormal{max}, j} = \max\left(\delta +(j-1)\delta', \delta', \delta'' \right)$ and $|\mathcal{M}_{max}| = \max_{i=1, 2,...,j} |\mathcal{M}_{i}|$.
                \begin{itemize}
                \item The first message from Bob to Alice is of fixed size $n\left( \delta_{s} \log_{2}(|\mathcal{Y}|)+\frac{1}{n}\log(n)+2\delta_{\textnormal{min}}^{2}+O\left(\frac{1}{n}\right)\right)$;
                    \item The second message from Alice to Bob is of variable length, but always of size at most 
                    
                    $ n\left(\delta_{s} \log_{2}\left(|\mathcal{X}| \right)+\log_{2}\left(|\mathcal{M}_{1}| \right)\right)$;
                    \item The $i^{th}$ message is of variable length, but always of size at most $n\cdot\log_{2}\left(|\mathcal{M}_{i-1}| \right)$;

                    \item There is an event $E_{good}$ such that 
                    \begin{equation}
                        \Pr[\lnot E_{good}] \leq j\cdot 2^{-n\cdot\delta_{\textnormal{min}}^{2}+3}
                    \end{equation}
                    and conditional on $E_{good}$:
                    \begin{itemize}
                    \item The total communication cost is at most
                    \begin{equation}
                        n(I(M_{1}...M_{j};X|Y)_{p\cdot t}
                        +I(M_{1}...M_{j};Y|X)_{p\cdot t}+\delta_{s} \log\left(|\mathcal{X}\right|\cdot\left|\mathcal{Y}\right|+\eta_{3}(n, \delta_{\textnormal{max}, j}, j));
                    \end{equation}
                    \item The amount of shared randomness needed is at most 
                    \begin{equation}
                        n\left(H(M_{1}...M_{j}|X, Y)_{p\cdot t}+\frac{j}{n}\log\log(e)\right);
                    \end{equation}
                    \item Alice and Bob each produce a copy of the same transcript of outputs $\hat{m}_{1}^{n}, ..., \hat{m}_{j}^{n}$;
                        \item It holds that for every $(x^{n}, y^{n}) \in \mathcal{X}^{n} \times \mathcal{Y}^{n}$, if $\widetilde{p^{n}}(m_{1}^{n}, ..., m_{j}^{n}|x^{n}, y^{n})$ is the probability that Alice and Bob obtain the transcript of outputs $m_{1}^{n}, ..., m_{j}^{n}$ from the protocol, then
                    \end{itemize}
                \end{itemize}
                \begin{align}
                           \left\|\widetilde{p^{n}}(m_{\leq j}^{n}|x^{n}, y^{n})-p^{n}(m_{\leq j}^{n}|x^{n}, y^{n})\right\|_{1}\leq  j\left(2^{-n\delta'''} + \delta''\right);
                        \end{align}
                where 
                \begin{align}
                \eta_{3}(n, \delta_{\textnormal{max}, j}, j) = 5j\gamma(|\mathcal{M}_{max}|, \left|\mathcal{X}\right|\cdot\left|\mathcal{Y}\right|\cdot\delta_{\textnormal{max}, j})+\frac{2j}{n}\log(n)
                +\frac{4j}{n}|\mathcal{M}_{\leq j}| \cdot|\mathcal{X}|\cdot |\mathcal{Y}|\log\left(n+1\right)\\+3j\delta_{\textnormal{max}, j}+O\left(\frac{j}{n}\right).
                \end{align}
        \end{theorem}
        \begin{proof}
            The following protocol simulates a $j$ round asymptotic, interactive communication protocol.
            \begin{enumerate}
                \item Alice and Bob perform step $1$ of Lemma~\ref{lem:estimating} to obtain the strings of symbols to exchange, with $m = n \cdot \delta_{s}$. Denote their strings $s_{A}$ and $s_{B}$. Bob transmits his string $s_{B}$. This concludes the first round of the protocol.
                    \item \label{step:est} Alice sends her string $s_{A}$. Both parties can now determine the estimated distribution $\widetilde{t}$.
                    \item \label{step:first}Alice and Bob perform the protocol of Theorem~\ref{thm:rst_1} using $\widetilde{t}$ to obtain a first simulated message $\hat{m}^{n}_{1}$.
                \item \label{step:rep} Alice and Bob perform the protocol of Theorem~\ref{thm:rst_1} with Bob as the sender, where they take Alice's input as $(x^{n}, \hat{m}^{n}_{1})$, Bob's input as $(y^{n}, \hat{m}^{n}_{1})$, the joint input distribution as $p_{M_{1}|X}\cdot\widetilde{t}_{X, Y}$, and the input parameter $\delta$ is updated to be $\delta+\delta'$, to generate a second message $\hat{m}^{n}_{2}$.
                \item \label{step:rep_2} Alice and Bob repeat step \ref{step:rep} with the appropriate inputs and joint distribution to obtain $j$ messages.
                \begin{itemize}
                    \item To generate the $ith$ message (round $i+1$) of the protocol using step \ref{step:rep}, the input parameter $\delta$ of Theorem~\ref{thm:rst_1} is updated at each iteration to be $\delta+(i-1)\delta'$.
                \end{itemize}
            \end{enumerate}
            For steps \ref{step:est} and \ref{step:first}, by Lemma~\ref{lem:estimating}, the probability that $(x^{n}, y^{n}) \not\in T^{X^{n}, Y^{n}}_{\overline{t}, \delta}$ is upper bounded by 
            \begin{equation}
                2^{-n\left(2\delta_{s}\delta^{2}\log e+\frac{1}{n}\log\left|\mathcal{X}\right|\cdot \left|\mathcal{Y}\right|+\frac{1}{n}\right)}.
            \end{equation}
            The probability that the protocol of Theorem~\ref{thm:rst_1} will fail without the de-randomization is at most $4\cdot2^{-n\cdot\delta_{\textnormal{min}}^{2}}$. By the union bound, the probability that either occurs is at most $5\cdot2^{-n\cdot\delta_{\textnormal{min}}^{2}}$. We can then apply a similar derandomization argument as in the proof of Claim~\ref{clm:derandom_rst}, which results in a probability of error of at most $2^{-n\cdot\delta_{\textnormal{min}}^{2}+3}$, and step~\ref{step:first} will have the same communication and shared randomness costs as in Theorem~\ref{thm:rst_1}.
            
            For the first iteration of step~\ref{step:rep}, if step~\ref{step:first} was successful, the message generated in that round $\hat{m}^{n}_{1}$ is in the set $T^{M_{1}^{n}, X^{n}, Y^{n}}_{(\widetilde{t}, \delta), (p_{1}, \delta')}.$
            By proposition $\ref{prop:merge_set}$, it must also hold that the message is in the set $
                T^{M_{1}^{n}, X^{n}, Y^{n}}_{(p_{1}\cdot\widetilde{t}, \delta+\delta')}.$ For the $ith$ iteration of step~\ref{step:rep} ($i>1$), if the previous iteration was successful, then the message generated in the previous iteration $\hat{m}^{n}_{i-1}$ is in the set $
                T^{M_{i-1}^{n}, X^{n}, Y^{n}, M^{n}_{<(i-1)}}_{(p_{i-2}\cdot p_{i-3}\cdot...p_{1}\cdot\widetilde{t}, \delta+(i-2)\delta'), (p_{i-1}, \delta')}$. By proposition $\ref{prop:merge_set}$, it must also hold that the message is in the set $
                T^{M_{i-1}^{n}, X^{n}, Y^{n}, M^{n}_{<(i-1)}}_{(p_{i-1}\cdot p_{i-2}\cdot...p_{1}\cdot\widetilde{t}, \delta+(i-1)\delta')}$.
            Therefore, the conditions to apply Theorem~\ref{thm:rst_1} hold.

            Let $E_{good}^{i}$ denote the event that the $i^{th}$ message was successfully generated ($(i+1)^{th}$ round was successful). We can then define $E_{good} = \bigwedge_{i} E_{good}^{i}$.
            By the union bound, the probability of the event $\lnot E_{good}$ is at most $j\cdot 2^{-n\cdot\delta_{\textnormal{min}}^{2}+3}$.

            The communication cost for the $i^{th}$ iteration of step~\ref{step:rep}, which generates the $i^{th}$ message of the protocol (assume $i$ is odd) is bound by $
                n\left(I\left(M;X|Y\right)_{p_{i}\cdot p_{i-1}\cdot...\cdot p_{1}\cdot t}+\eta_{2}(n, \delta_{\textnormal{max}, i})\right)$,
            where $\delta_{\textnormal{max}, i} = \max\left(\delta+(i-1)\delta', \delta', \delta'' \right)$. Note that the parameter corresponding to $\delta+(i-1)\delta'$ is the only one that needs to be updated from one round to the next. The parameters $\delta', \delta'', \delta_{s}$ and consequently $\delta'''$ remain unchanged.

            The total communication cost of simulating all $j$ rounds is obtained by summing the cost over all $j$ rounds. A detailed analysis of the communication cost and the total variation distance between the true and simulated channels can be found in the \hyperref[proof:int_2]{Appendix}.
        \end{proof}
\subsection{$j$-round protocol simulating a $j$-round interactive protocol without an estimate of the joint type given}
        In the following theorem, we adapt Theorem~\ref{thm:int_2} to show that the simulation can be performed in $j$ rounds at the expense of increasing the communication rate of the first round of the simulation. 
        \begin{theorem}\label{thm:int_3}
            Let $\mathcal{X}$, $\mathcal{Y}$ and $\mathcal{M}_{\leq j} = \mathcal{M}_{1}\times...\times\mathcal{M}_{j}$ be finite alphabets, and fix an $n \in \mathbb{N}$. There exists a $j$-round protocol taking $\left(x^{n}, y^{n}\right) \in \mathcal{X}^{n} \times \mathcal{Y}^{n}$ as Alice and Bob's respective inputs, as well as parameters $\delta, \delta', \delta'', \delta_{s} > 0$, a description of $j$ noisy channels $N_{1}, ..., N_{j}$, each associated with an output distribution $p_{M_{i}|X, M_{< i}}$ or $p_{M_{i}|Y, M_{< i}}$ depending on the parity of $i$, as additional shared inputs, that satisfies the following conditions:

            Define $\delta''' = \left(\frac{\delta'^{2}}{2\ln 2}-\frac{2}{n}|\mathcal{M}|\cdot|\mathcal{X}|\cdot|\mathcal{Y}|\log(n+1)\right)$, $\delta_{\textnormal{min}} = \min\left(\delta, \delta''', \delta'' \right)$ and 
            
            $\delta_{\textnormal{max}, j} = \max\left(\delta +(j-1)\delta', \delta', \delta'' \right)$.
                \begin{itemize}
                    \item The first message from Alice to Bob is of variable length, but always of size at most 
                    
                    ${n\left( \delta_{s} \log_{2}(|\mathcal{X}|)+\log_{2}\left(|\mathcal{M}_{1}| \right)+\frac{1}{n}\log(n)+2\delta_{\textnormal{min}}^{2}+O\left(\frac{1}{n}\right)\right)}$;
                    \item The second message from Bob to Alice is of variable length, but always of size at most
                    
                    ${n\left(\delta_{s} \log_{2}\left(|\mathcal{Y}| \right)+\log_{2}\left(|\mathcal{M}_{2}|\right)\right)}$;
                    \item The $i^{th}$ message is of variable length, but always of size at most $n\cdot\log_{2}\left(|\mathcal{M}_{i}| \right)$;
                    \item There is an event $E_{good}$ such that $\Pr[\lnot E_{good}] \leq j\cdot 2^{-n\cdot\delta_{\textnormal{min}}^{2}+3}$, such that conditional on $E_{good}$:
                    \begin{itemize}
                    \item The communication cost is at most 
                    \begin{align}
                        n\bigg(\max_{x^{n}, y^{n}\in \mathcal{X}^{n}\times\mathcal{Y}^{n}} \left[I(M_{1};X|Y)_{p\cdot t}\right]
                        +\max_{x^{n}, y^{n}} \left[I(M_{2}...M_{j};X|M_{1}Y)_{p\cdot t}
                        +I(M_{2}...M_{j};Y|M_{1}X)_{p\cdot t}\right]\\
                        + \delta_{s} \log\left(\left|\mathcal{X}\right|\cdot\left|\mathcal{Y}\right|\right)+\eta_{3}(n, \delta_{\textnormal{max}, j}, j)\bigg);
                    \end{align}
                    \item The amount of shared randomness needed is at most $n\left(H(M_{1}...M_{j}|X, Y)_{p\cdot t}+\frac{j}{n}\log\log(e)\right)$;
                    \item Alice and Bob each produce a copy of the same transcript of outputs $\hat{m}_{1}^{n}, ..., \hat{m}_{j}^{n}$;
                        \item It holds that for every $(x^{n}, y^{n}) \in \mathcal{X}^{n} \times \mathcal{Y}^{n}$, if $\widetilde{p^{n}}(\hat{m}_{1}^{n}, ..., \hat{m}_{j}^{n}|x^{n}, y^{n})$ is the probability that Alice and Bob obtain the output $m^{n}$ from the protocol, then
                        \begin{equation}
                           \left\|\widetilde{p^{n}}(\hat{m}_{1}^{n}, ..., \hat{m}_{j}^{n}|x^{n}, y^{n})-p^{n}(\hat{m}_{1}^{n}, ..., \hat{m}_{j}^{n}|x^{n}, y^{n})\right\|_{1}
                           \leq  j\left(2^{-n\delta'''} + \delta''\right);
                        \end{equation}
                    \end{itemize}
                \end{itemize}
                where $\delta''$ and $\eta_{4}$ are defined as in Theorem~\ref{thm:int_2}.
        \end{theorem}
        \begin{proof}
        The following protocol simulates $j$ rounds of $n$ uses of a noisy channel.
            \begin{enumerate}
        \item Alice privately samples $\log n+2n\cdot\delta_{\textnormal{min}}^{2} +O(1)$ bits of randomness and communicates it to Bob. This will be used to perform the de-randomized versions of Lemma~\ref{lem:estimating} and Theorem~\ref{thm:rst_1}.
                \item Alice and Bob perform step $1$ of Lemma~\ref{lem:estimating} to obtain the strings of symbols to exchange. Denote their strings $s_{A}$ and $s_{B}$.
                \item Alice and Bob perform the protocol of Theorem~\ref{thm:rst_1} assuming $\widetilde{t}$ is a product distribution on $\mathcal{X} \times \mathcal{Y}$ to obtain a first simulated message $\hat{m}^{n}_{1}$. 
                \item Alice communicates $s_{A}$ to Bob. With this they complete round $1$ of the protocol.
                \item Bob communicates $s_{B}$ to Alice and they can then both perform step $3$ of Lemma~\ref{lem:estimating} and obtain $\widetilde{t}$.
                \item Alice and Bob perform steps $\ref{step:rep}$ and $\ref{step:rep_2}$ of Theorem~\ref{thm:int_2}.
            \end{enumerate}
            The communication cost of steps $2-4$ is at most
            \begin{equation}
               n\left(I\left(M_{1};X\right)_{ p\cdot t}+\delta_{s}  \log(|\mathcal{X}|\cdot |\mathcal{Y}|)+\eta_{2}(n, \delta_{\textnormal{max}})\right).
            \end{equation}    
            Let $x^{n}$ be the sequence that maximizes $I(M_{1};X)_{p\cdot t}$, let $\overline{y^{n}}$ be a sequence such that $t_{x^{n}, \overline{y^{n}}} = t_{x^{n}}\cdot t_{\overline{y^{n}}}$, and define $X\overline{Y}M_{1} \sim p_{M_{1}|X}\cdot t_{x^{n}, \overline{y^{n}}}$, then
            \begin{align}
                \max_{x^{n}\in \mathcal{X}^{n}}I(M_{1};X)_{p_{M_{1}|X}} &= \max_{x^{n}\in \mathcal{X}^{n}}I(M_{1};X|\overline{Y})_{p_{M_{1}|X}\cdot t_{x^{n}, \overline{y^{n}}}}\\
                &\leq \max_{x^{n}, y^{n}\in \mathcal{X}^{n}\times \mathcal{Y}^{n}} I(M_{1};X|Y)_{p_{M_{1}|X}\cdot t_{x^{n}, y^{n}}}.
            \end{align}
            Therefore, the maximum communication cost of steps $2-4$ is upper bounded by
            \begin{equation}
                n\bigg(\max_{x^{n}, y^{n}\in \mathcal{X}^{n}\times \mathcal{Y}^{n}}\left[I(M_{1};X|Y)_{p_{M_{1}|X}\cdot t_{x^{n}, y^{n}}}\right]+\delta_{s}  \log(|\mathcal{X}|\cdot |\mathcal{Y}|)
               +\eta_{2}(n, \delta_{\textnormal{max}})\bigg).
            \end{equation}    
            Theorem \ref{thm:int_2} covers the analysis of the communication cost of steps $5$ and $6$. Therefore, for any $x^{n}$ and $y^{n}$, we obtain the following upper bound on the communication cost.
            \begin{align}
                n \bigg(\max_{x^{n}, y^{n}\in \mathcal{X}^{n}\times\mathcal{Y}^{n}} \left[I(M_{1};X|Y)_{p_{1}\cdot t}\right]+\max_{x^{n}, y^{n}} [I(M_{2}...M_{j};X|M_{1}Y)_{p\cdot t}+I(M_{2}...M_{j};Y|M_{1}X)_{p \cdot t}]\\+\delta_{s}  \log(|\mathcal{X}|\cdot |\mathcal{Y}|)
                +\eta_{3}(n, \delta_{\textnormal{max}, j}, j) \bigg).
            \end{align}
        \end{proof}

    %%%%%

    \section{Conclusion} \label{sec:conclusion}
    We have demonstrated that the amortized communication complexity of simulating a prior-free interactive communication protocol with round preservation is equal to its information cost. We developed a prior-free reverse Shannon theorem with side information at the receiver that uses a sublinear amount of noiseless communication to sample and estimate the joint empirical distribution of the input sequences. Our protocols lead to a simpler proof of the results of \cite{Braverman2017}, and improve on the results by achieving round preservation, and obtaining optimal bounds on the amount of shared randomness required. A direction of further development would be to extend these results to the quantum setting. Quantum computation brings with it promises of faster computation as well as many open questions that would benefit from the further development of the field of quantum communication complexity. The results obtained here are a first step in demonstrating the equivalence between the prior-free quantum information and amortized quantum communication complexity of an interactive protocol. 

    \section{Acknowledgements}
        We would like to thank Professor Claude Crépeau for his supervision and guidance throughout this work. We acknowledge the support of the Natural Sciences and Engineering Research Council of Canada (NSERC)(ALLRP-578455-2022 and DG), the NSERC‐Collaborative Research and Training Experience program QSciTech, the Institut transdisciplinaire d'information quantique (INTRIQ), and the support of the Ministère de l'Enseignement Supérieur du Québec.
    
    %%%%%%
    %% References:
    %% We recommend the usage of BibTeX:
    %%
    \newpage
    \bibliographystyle{IEEEtran}
    \bibliography{bibliofile}

    \clearpage
    
     \appendix

        \onecolumn
        \section{Proofs}
		\label{App:proofs}
    %%%%%%
            \begin{proof}[Proof of Proposition~\ref{prop:card_joint_typ_set}]
            \label{proof:card_joint_typ_set}
                            \begin{align}
                                \left|T^{X^{n}, Y^{n}}_{\widetilde{t}, \delta} \right|&= \left|\bigcup_{\substack{t \in \mathcal{T}^{\mathcal{X}^{n}, \mathcal{Y}^{n}}\\ \|t-\widetilde{t}\|_{1} \leq \delta}}T^{X^{n},Y^{n}}_{t}\right| \\
                                &= \sum_{\substack{t \in \mathcal{T}^{\mathcal{X}^{n}, \mathcal{Y}^{n}}\\ \|t-\widetilde{t}\|_{1} \leq \delta}}\left|T^{X^{n},Y^{n}}_{t}\right|  \\
                                \sum_{\substack{t \in \mathcal{T}^{\mathcal{X}^{n}, \mathcal{Y}^{n}}\\ \|t-\widetilde{t}\|_{1} \leq \delta}} (n+1)^{-\left|\mathcal{X}\right|\cdot\left|\mathcal{Y}\right|}2^{nH(X, Y)_{t}} \leq \left|T^{X^{n}, Y^{n}}_{\widetilde{t}, \delta} \right| &\leq  \sum_{\substack{t \in \mathcal{T}^{\mathcal{X}^{n}, \mathcal{Y}^{n}}\\ \|t-\widetilde{t}\|_{1} \leq \delta}}2^{nH(X, Y)_{t}}   \\
                                \sum_{\substack{t \in \mathcal{T}^{\mathcal{X}^{n}, \mathcal{Y}^{n}}\\ \|t-\widetilde{t}\|_{1} \leq \delta}} (n+1)^{-\left|\mathcal{X}\right|\cdot\left|\mathcal{Y}\right|}2^{n\left(H(X, Y)_{\widetilde{t}}-\gamma\left( \left|\mathcal{X}\right|\cdot \left| \mathcal{Y}\right|-1, \delta \right)\right)}    \leq \left|T^{X^{n}, Y^{n}}_{\widetilde{t}, \delta} \right| &\leq  \sum_{\substack{t \in \mathcal{T}^{\mathcal{X}^{n}, \mathcal{Y}^{n}}\\ \|t-\widetilde{t}\|_{1} \leq \delta}}2^{n\left(H(X, Y)_{\widetilde{t}}+\gamma\left( \left|\mathcal{X}\right|\cdot \left| \mathcal{Y}\right|-1, \delta \right)\right)}   \\
                                (n+1)^{-\left|\mathcal{X}\right|\cdot\left|\mathcal{Y}\right|}2^{n\left(H(X, Y)_{\widetilde{t}}-\gamma\left( \left|\mathcal{X}\right|\cdot \left| \mathcal{Y}\right|-1, \delta \right)\right)} \leq \left|T^{X^{n}, Y^{n}}_{\widetilde{t}, \delta} \right| &\leq  \left|\mathcal{T}^{\mathcal{X}^{n}, \mathcal{Y}^{n}}\right|2^{n\left(H(X, Y)_{\widetilde{t}}+\gamma\left( \left|\mathcal{X}\right|\cdot \left| \mathcal{Y}\right|-1, \delta \right)\right)}   \\
                                (n+1)^{-\left|\mathcal{X}\right|\cdot\left|\mathcal{Y}\right|}2^{n\left(H(X, Y)_{\widetilde{t}}-\gamma\left( \left|\mathcal{X}\right|\cdot \left| \mathcal{Y}\right|-1, \delta \right)\right)} \leq \left|T^{X^{n}, Y^{n}}_{\widetilde{t}, \delta} \right|&\leq  (n+1)^{\left|\mathcal{X}\right|\cdot \left|\mathcal{Y}\right|}2^{n\left(H(X, Y)_{\widetilde{t}}+\gamma\left( \left|\mathcal{X}\right|\cdot \left| \mathcal{Y}\right|-1, \delta \right)\right)} \\
                                2^{n\left(H(X, Y)_{\widetilde{t}}-\gamma\left( \left|\mathcal{X}\right|\cdot \left| \mathcal{Y}\right|-1, \delta \right)\right)-\left|\mathcal{X}\right|\cdot \left|\mathcal{Y}\right|\log_{2}\left(n+1\right)} \leq \left|T^{X^{n}, Y^{n}}_{\widetilde{t}, \delta} \right|&\leq  2^{n\left(H(X, Y)_{\widetilde{t}}+\gamma\left( \left|\mathcal{X}\right|\cdot \left| \mathcal{Y}\right|-1, \delta \right)\right)+\left|\mathcal{X}\right|\cdot \left|\mathcal{Y}\right|\log_{2}\left(n+1\right)}.
                            \end{align}
                        \end{proof}
            \begin{proof}[Proof of Lemma~\ref{lem:estimating}]\label{proof:estimating}
                To determine the probability of achieving a reliable estimate, consider a pair of symbols $\left(x, y\right) \in \mathcal{X} \times \mathcal{Y}$. It holds that
                $\mathbb{E}\left[c\left(x, y\right)\right] = m \cdot t_{x^{n}, y^{n}}\left(x, y\right).$ By the Chernoff-Hoeffding bound \cite{hoeffding_1963}, for $0< \delta' \leq 1-t_{x^{n}, y^{n}}\left(x, y\right)$ :
                \begin{equation}
                    \Pr\left[\left|\frac{c\left(x, y\right)}{m} -  t_{x^{n}, y^{n}}\left(x, y\right)\right|\geq \delta'\right] \leq 2e^{-2m\delta'^{2}}.
                \end{equation}
                We then have
                \begin{align}
                    &\Pr\left[\forall (x, y) \in \left(\mathcal{X}\times \mathcal{Y}\right): \left|\overline{t}\left(x, y\right)-t_{x^{n}, y^{n}}\left(x, y\right)\right| < \delta\right]\\
                    &= 1-\Pr\left[\exists (x, y) \in \left(\mathcal{X}\times \mathcal{Y}\right): \left|\overline{t}\left(x, y\right)-t_{x^{n}, y^{n}}\left(x, y\right)\right| \geq \delta\right]\\
                    &\leq 1- |\mathcal{X}|\cdot|\mathcal{Y}|2e^{-2m\delta'^{2}}\\
                    &= 1- 2^{-2m\delta'^{2}\log e+\log|\mathcal{X}|\cdot|\mathcal{Y}|+1}.
                \end{align}
                Taking $\delta' = \delta$, it holds with probability at least $1- 2^{-2m\delta^{2}\log e+\log|\mathcal{X}|\cdot|\mathcal{Y}|+1}$ that for every $x, y \in \mathcal{X} \times \mathcal{Y}$ 
                \begin{equation}
                    \left| \widetilde{t}\left(x, y\right)-  t_{x^{n}, y^{n}}\left(x, y\right) \right| \leq \delta.
                \end{equation}
            \end{proof}
            
\begin{proof}[Proof of Claim~\ref{clm:marg_cond_typ_class}]\label{proof:marg_cond_typ_class}
    For every $x^{n} \in T^{X^{n}\backslash Y^{n}}_{\widetilde{t}}$, there must exist a $y^{n} \in \mathcal{Y}^{n}$ such that
    \begin{equation}
        \left(x^{n}, y^{n}\right) \in T^{X^{n}, Y^{n}}_{\widetilde{t}}.
    \end{equation}
    Therefore, the joint type of $(x^{n}, y^{n})$ is exactly $\widetilde{t}$, and more specifically, the type of $x^{n}$ is the $\mathcal{X}$ marginal of the type $\widetilde{t}$. Therefore, it must hold that $x^{n} \in T^{X^{n}}_{\widetilde{t}}$.

    For every $x^{n} \in T^{X^{n}}_{\widetilde{t}}$, since the type of $x^{n}$ is exactly the $\mathcal{X}$ marginal of the type $\widetilde{t}$, there must exist a sequence $y^{n}\in \mathcal{Y}^{n}$ such that
    \begin{equation}
        \left(x^{n}, y^{n}\right) \in T^{X^{n}, Y^{n}}_{\widetilde{t}}.
    \end{equation}
    Therefore, it must hold that $x^{n} \in T^{X^{n}\backslash Y^{n}}_{\widetilde{t}}$.
\end{proof}

\begin{proof}[Proof of Claim~\ref{clm:cond_typ_class}]\label{proof:cond_typ_class}
  Let $y^{n}$ be a sequence with empirical distribution $\widetilde{t}$. Any other sequence $\widetilde{y^{n}}$ with the same empirical distribution as $y^{n}$ can be obtained by permuting the symbols of $y^{n}$. The same permutation can then be applied to every sequence in the set $T^{X^{n}|y^{n}}_{\widetilde{t}}$ to obtain the set $T^{X^{n}|\widetilde{y^{n}}}_{ \widetilde{t}}$. It must then hold that $\left|T^{X^{n}|y^{n}}_{\widetilde{t}}\right| = \left|T^{X^{n}|\widetilde{y^{n}}}_{\widetilde{t}}\right|$ for every pair of sequences $y^{n}, \widetilde{y^{n}} \in T^{Y^{n}}_{\widetilde{t}}$. The joint type class can be partitioned according to the $\mathcal{Y}^{n}$ sequences
    \begin{equation}
        T_{\widetilde{t}} ^{X^{n}, Y^{n}} = \bigcup_{y^{n} \in T^{Y^{n}}_{\widetilde{t}}}T^{X^{n}| y^{n}}_{\widetilde{t}} \times \{y^{n}\}.
    \end{equation}
    This is a disjoint union, since even if the set $T^{X^{n}|y^{n}}_{\widetilde{t}}$ is invariant under the permutation that takes $y^{n}$ to $\widetilde{y^{n}}$, the sequences $y^{n}$ and $\widetilde{y^{n}}$ are distinct, and so the sets $T^{X^{n}| y^{n}}_{\widetilde{t}} \times \left\{y^{n}\right\}$ and $T^{X^{n}| \widetilde{y^{n}}}_{\widetilde{t}} \times \left\{\widetilde{y^{n}}\right\}$ will be disjoint. We therefore have
\begin{align}
    \left|T^{X^{n}, Y^{n}}_{\widetilde{t}}\right| &= \sum_{y^{n} \in T^{Y^{n}}_{\widetilde{t}}}\left|T^{X^{n}|y^{n}}_{\widetilde{t}}\right|\\
    &=\left|T^{Y^{n}}_{\widetilde{t}}\right| \cdot\left|T^{X^{n}|y^{n}}_{\widetilde{t}}\right|.
\end{align}
\end{proof}
\begin{proof}[Proof of Proposition \ref{prop:card_cond_typ_set}]
            \label{proof:card_cond_typ_set}
                            \begin{align}
                                \left|T^{X^{n}|y^{n}}_{ \widetilde{t}, \delta}\right|&= \left|\bigcup_{\substack{t \in \mathcal{T}^{\mathcal{X}^{n}, \mathcal{Y}^{n}}\\ \|t-\widetilde{t}\|_{1} \leq \delta}}T^{X^{n} | y^{n}}_{t}\right| \\
                                &= \sum_{\substack{t \in \mathcal{T}^{\mathcal{X}^{n}, \mathcal{Y}^{n}}\\ \|t-\widetilde{t}\|_{1} \leq \delta}}\left|T^{X^{n} | y^{n}}_{t}\right|  \\
                                \sum_{\substack{t \in \mathcal{T}^{\mathcal{X}^{n}, \mathcal{Y}^{n}}\\ \|t-\widetilde{t}\|_{1} \leq \delta}} 2^{nH\left(X| Y\right)_{t}-|\mathcal{X}| \cdot |\mathcal{Y}|\log\left(n+1\right)}   \leq \left|T^{X^{n}|y^{n}}_{ \widetilde{t}, \delta} \right| &\leq  \sum_{\substack{t \in \mathcal{T}^{\mathcal{X}^{n}, \mathcal{Y}^{n}}\\ \|t-\widetilde{t}\|_{1} \leq \delta}}2^{nH\left(X| Y\right)_{t}+|\mathcal{X}|\cdot|\mathcal{Y}|\log\left(n+1\right)}  \\
                                \sum_{\substack{t \in \mathcal{T}^{\mathcal{X}^{n}, \mathcal{Y}^{n}}\\ \|t-\widetilde{t}\|_{1} \leq \delta}} 2^{n(H\left(X| Y\right)_{\widetilde{t}}-\gamma\left( \left|\mathcal{X}\right|, \delta \right))-|\mathcal{X}| \cdot |\mathcal{Y}|\log\left(n+1\right)}    \leq \left|T^{X^{n}|y^{n}}_{ \widetilde{t}, \delta} \right| &\leq  \sum_{\substack{t \in \mathcal{T}^{\mathcal{X}^{n}, \mathcal{Y}^{n}}\\ \|t-\widetilde{t}\|_{1} \leq \delta}}2^{n\left(H(X|Y)_{\widetilde{t}}+\gamma\left( \left|\mathcal{X}\right|, \delta \right)\right)+|\mathcal{X}|\cdot|\mathcal{Y}|\log\left(n+1\right)}   \\
                                2^{n(H\left(X| Y\right)_{\widetilde{t}}-\gamma\left( \left|\mathcal{X}\right|, \delta \right))-|\mathcal{X}| \cdot |\mathcal{Y}|\log\left(n+1\right)} \leq \left|T^{X^{n}|y^{n}}_{ \widetilde{t}, \delta}\right| &\leq  \left|\mathcal{T}^{\mathcal{X}^{n}, \mathcal{Y}^{n}}\right|2^{n\left(H(X|Y)_{\widetilde{t}}+\gamma\left( \left|\mathcal{X}\right|, \delta \right)\right)+|\mathcal{X}|\cdot|\mathcal{Y}|\log\left(n+1\right)}  \\
                                &\leq  (n+1)^{\left|\mathcal{X}\right|\cdot \left|\mathcal{Y}\right|}2^{n\left(H(X|Y)_{\widetilde{t}}+\gamma\left( \left|\mathcal{X}\right|, \delta \right)\right)+|\mathcal{X}|\cdot|\mathcal{Y}|\log\left(n+1\right)} \\
                                &=  2^{n\left(H(X|Y)_{\widetilde{t}}+\gamma\left( \left|\mathcal{X}\right|, \delta \right)\right)+2\left|\mathcal{X}\right|\cdot|\mathcal{Y}|\log\left(n+1\right)}.
                            \end{align}
                        \end{proof}
            \begin{proof}[Proof of Proposition~\ref{prop:unit_prob}]
            \label{proof:unit_prob}
    \begin{align}
        p^{n}(m^{n}|x^{n}, y^{n}) &= \prod_{i=1}^{n} p(m_{i}|x_{i}, y_{i})\\
        &= \prod_{(m, x, y) \in \mathcal{M}\times \mathcal{X} \times \mathcal{Y}} p(m|x, y)^{nt_{m^{n}, x^{n}, y^{n}}(m ,x, y)}\\
        &= \prod_{(m, x, y) \in \mathcal{M}\times \mathcal{X} \times \mathcal{Y}} p(m|x, y)^{t_{m^{n}|x^{n}, y^{n}}(m|x, y)\cdot t_{x^{n}, y^{n}}(x, y)}\\
        &=  \prod_{(m, x, y) \in \mathcal{M}\times \mathcal{X} \times \mathcal{Y}} \left(2^{nt_{m^{n}| x^{n}, y^{n}}(m|x, y)\log_{2}p(m|x, y)}\right)^{t_{x^{n}, y^{n}}(x, y)}\\
        &=  \prod_{(m, x, y) \in \mathcal{M}\times \mathcal{X} \times \mathcal{Y}} \left(2^{n\left(t_{m^{n}| x^{n}, y^{n}}(m|x, y)\log_{2}\left(p(m|x, y)\right)\pm t_{m^{n}| x^{n}, y^{n}}(m|x, y)\log_{2}\left(t_{m^{n}| x^{n}, y^{n}}(m|x, y)\right)\right)}\right)^{t_{x^{n}, y^{n}}(x, y)}\\
         &= 2^{-nH(M|X, Y)_{t}}\prod_{(x, y)\in \mathcal{X}\times \mathcal{Y}} \left(2^{-n\left(D\left(t_{m^{n}|x^{n}, y^{n}}||p\right)\right)}\right)^{t_{x^{n}, y^{n}}(x, y)}\\
         &= 2^{-nH(M|X, Y)_{t}} 2^{-n\sum_{(x, y)\in \mathcal{X}\times\mathcal{Y}} D\left(t_{m^{n}|x^{n}, y^{n}}||p\right)t_{x^{n}, y^{n}}(x, y)}\\
         &= 2^{-nH(M|X, Y)_{t}} 2^{-nD\left(t_{m^{n}|x^{n}, y^{n}}||p\right)\cancelto{1}{\sum_{(x, y)\in \mathcal{X}\times \mathcal{Y}}t_{x^{n}, y^{n}}(x, y)}}
         \end{align}
         \begin{align}
         p^{n}\left(T^{M^{n}|x^{n}, y^{n}}_{t}\right)&= \sum_{m^{n} \in T^{M^{n}|x^{n}, y^{n}}_{t}}2^{-nH(M|X, Y)_{t}} 2^{-nD\left(t_{m^{n}|x^{n}, y^{n}}||p\right)} \\
         &=\left|T^{M^{n}|x^{n}, y^{n}}_{t} \right|2^{-n\left(H(M|X, Y)_{t}\right)} 2^{-nD\left(t_{m^{n}|x^{n}, y^{n}}||p\right)}\\
         &\leq 2^{n\cancel{H\left(M|X, Y\right)_{t}}+|\mathcal{X}|\cdot|\mathcal{Y}|\log\left(n+1\right)}2^{-n\left(\cancel{H(M|X, Y)_{t}}\right)} 2^{-nD\left(t_{m^{n}|x^{n}, y^{n}}||p\right)}\\
         &= 2^{|\mathcal{X}|\cdot|\mathcal{Y}|\log\left(n+1\right)} 2^{-nD\left(t_{m^{n}|x^{n}, y^{n}}||p\right)}
         \end{align}
         \begin{align}
        p^{n}\left(T^{M^{n}|x^{n}, y^{n}}_{\widetilde{t}, (p,\delta')}\right) &=p^{n} \left(\bigcup_{\substack{t\in \mathcal{T}^{\mathcal{M}^{n}, \mathcal{X}^{n}, \mathcal{Y}^{n}}\\ \|t -p\cdot\widetilde{t}\|_{1} \leq \delta'}}T^{M^{n} |x^{n}, y^{n}}_{t}\right)\\
        &=\sum_{\substack{t\in \mathcal{T}^{\mathcal{M}^{n}, \mathcal{X}^{n}, \mathcal{Y}^{n}}\\ \|t -p\cdot\widetilde{t}\|_{1} \leq \delta'}}p^{n} \left(T^{M^{n} |x^{n}, y^{n}}_{t}\right)\\
        &=1-\sum_{\substack{t\in \mathcal{T}^{\mathcal{M}^{n}, \mathcal{X}^{n}, \mathcal{Y}^{n}}\\ \|t -p\cdot\widetilde{t}\|_{1} > \delta'}}p^{n} \left(T^{M^{n} |x^{n}, y^{n}}_{t}\right)\\
        &\geq 1-\sum_{\substack{t\in \mathcal{T}^{\mathcal{M}^{n}, \mathcal{X}^{n}, \mathcal{Y}^{n}}\\ \|t -p\cdot\widetilde{t}\|_{1} >\delta'}}2^{|\mathcal{X}|\cdot|\mathcal{Y}|\log\left(n+1\right)} 2^{-nD\left(t_{m^{n}|x^{n}, y^{n}}||p\right)}\\
        \intertext{Using Pinsker's inequality} &\geq 1-\sum_{\substack{t\in \mathcal{T}^{\mathcal{M}^{n}, \mathcal{X}^{n}, \mathcal{Y}^{n}}\\ \|t -p\cdot\widetilde{t}\|_{1} > \delta'}}2^{|\mathcal{X}|\cdot|\mathcal{Y}|\log\left(n+1\right)} 2^{-\frac{n\delta'^{2}}{2\ln 2}}\\
        &\geq 1-(n+1)^{2|\mathcal{M}|\cdot|\mathcal{X}|\cdot|\mathcal{Y}|}2^{-\frac{n\delta'^{2}}{2\ln 2}}\\
        &\geq 1-2^{-\frac{n\delta'^{2}}{2\ln 2}+2|\mathcal{M}|\cdot|\mathcal{X}|\cdot|\mathcal{Y}|\log(n+1)}.
    \end{align}
\end{proof}
            \begin{proof}[Proof of Proposition \ref{prop:card_cond_bs_set}]
            \label{proof:card_cond_bs_set}
        Consider the set defined as
        \begin{align}
        \label{eqn:temp_set}
            T^{M^{n}|y^{n}}_{\widetilde{t}, (p, \left|\mathcal{X}\right|\cdot\delta')} &= \bigcup_{\substack{t \in \mathcal{T}^{\mathcal{M}^{n}, \mathcal{Y}^{n}}\\
            \|t-\sum_{x\in\mathcal{X}}p\cdot \widetilde{t}\|_{1} \leq \left|\mathcal{X}\right|\cdot\delta'}} T^{M^{n}|y^{n}}_{t},\\
            \intertext{where $T^{M^{n}|y^{n}}_{t}$ is defined as}
            T^{M^{n}|y^{n}}_{t} &= \left\{m^{n}:(m^{n}, y^{n}) \in T^{M^{n}, Y^{n}}_{t}\right\}.
            \end{align}
            We will first demonstrate that $\left|T^{M^{n}|y^{n}\backslash X^{n}}_{(\widetilde{t}, \delta), (p,  \delta')}\right| \leq \left|T^{M^{n}|y^{n}}_{\widetilde{t}, (p, \left|\mathcal{X}\right|\cdot\delta')}\right|$ by showing that every $m^{n} \in T^{M^{n}|y^{n}\backslash X^{n}}_{(\widetilde{t}, \delta), (p,  \delta')}$ is also in $T^{M^{n}|y^{n}}_{\widetilde{t}, (p, \left|\mathcal{X}\right|\cdot\delta')}$. If $m^{n} \in T^{M^{n}|y^{n}\backslash X^{n}}_{(\widetilde{t}, \delta), (p,  \delta')}$, then there exists an $x^{n} \in \mathcal{X}^{n}$ such that $(m^{n}, x^{n}, y^{n}) \in T^{M^{n}, X^{n}, Y^{n}}_{(\widetilde{t}, \delta), (p, \delta')}$ which implies that
            \begin{align}
            \left\|t_{m^{n}, x^{n}, y^{n}} - p_{M|X}\cdot\widetilde{t}\right\|_{1} &\leq \delta'\\
            \sum_{x\in\mathcal{X}}\left\|t_{m^{n}, x^{n}, y^{n}} - p_{M|X}\cdot\widetilde{t}\right\|_{1} &\leq \left|\mathcal{X}\right|\cdot\delta'\\
            \left\|t_{m^{n}, y^{n}} - \sum_{x\in\mathcal{X}}p_{M|X}\cdot\widetilde{t}\right\|_{1} &\leq \left|\mathcal{X}\right|\cdot\delta'.
            \end{align}
            There is therefore a type, $t_{m^{n}, y^{n}}$ in particular, that satisfies the criteria allowing $T^{M^{n}|y^{n}}_{t_{m^{n}, y^{n}}}$ to be included in the union of (\ref{eqn:temp_set}), and therefore $m^{n} \in T^{M^{n}|y^{n}}_{\widetilde{t}, (p, \left|\mathcal{X}\right|\cdot\delta')}$. To bound $\left|T^{M^{n}|y^{n}}_{\widetilde{t}, (p, \left|\mathcal{X}\right|\cdot\delta')}\right|$,
            \begin{align}
             \left|T^{M^{n}|y^{n}}_{\widetilde{t}, (p, \left|\mathcal{X}\right|\cdot\delta')}\right| &\leq \sum_{\substack{t \in \mathcal{T}^{\mathcal{M}^{n}, \mathcal{Y}^{n}}\\
            \|t-\sum_{x\in\mathcal{X}}p_{M|X}\cdot \widetilde{t}\|_{1} \leq \left|\mathcal{X}\right|\cdot\delta'}} \left|T^{M^{n}|y^{n}}_{t}\right|\\
            &\leq \sum_{\substack{t \in \mathcal{T}^{\mathcal{M}^{n}, \mathcal{Y}^{n}}\\
            \|t-\sum_{x\in\mathcal{X}}p_{M|X}\cdot \widetilde{t}\|_{1} \leq \left|\mathcal{X}\right|\cdot\delta'}} 2^{nH\left(M| Y\right)_{t}+|\mathcal{M}|\cdot |\mathcal{Y}|\log\left(n+1\right)}\\
            &\leq (n+1)^{|\mathcal{M}|\cdot|\mathcal{Y}|}2^{n\left(H\left(M| Y\right)_{\sum_{x\in\mathcal{X}}p_{M|X}\cdot \widetilde{t}}+\gamma(|\mathcal{M}|, \left|\mathcal{X}\right|\cdot\delta')\right)+|\mathcal{M}|\cdot |\mathcal{Y}|\log\left(n+1\right)}\\
            &=2^{n\left(H\left(M| Y\right)_{\sum_{x\in\mathcal{X}}p_{M|X}\cdot \widetilde{t}}+\gamma(|\mathcal{M}|, \left|\mathcal{X}\right|\cdot\delta')\right)+2|\mathcal{M}|\cdot |\mathcal{Y}|\log\left(n+1\right)}.
        \end{align}
        \end{proof}

                    \begin{proof}[Proof of Proposition~\ref{prop:merge_set}]
            \label{proof:merge_set}
    For $(m^{n}, x^{n}, y^{n}) \in T^{M^{n}, X^{n}, Y^{n}}_{(\widetilde{t}, \delta), (p, \delta')}$, it must hold that
    \begin{align}
       \left\|t_{x^{n}, y^{n}} -\widetilde{t}\right\|_{1} &\leq \delta\\
    \intertext{It must also hold that $m^{n} \in T^{M^{n}|x^{n}, y^{n}}_{t_{m^{n}, x^{n}, y^{n}}, (p, \delta')}$, which implies that}
        \left\|t_{m^{n}, x^{n}, y^{n}}-p\cdot t_{x^{n}, y^{n}}\right\|_{1} &\leq \delta'
        \end{align}
   We would like to satisfy the necessary condition for $(m^{n}, x^{n}, y^{n})$ to be in $T^{M^{n}, X^{n}, Y^{n}}_{(p \cdot \widetilde{t}, \delta'+\delta)}$. Using the triangle inequality, we have that
   \begin{align}
    \left\|t_{m^{n}, x^{n}, y^{n}}-p\cdot\widetilde{t}\right\|_{1}&=\left\|t_{m^{n}, x^{n}, y^{n}}-p\cdot\widetilde{t}\pm p\cdot t_{x^{n}, y^{n}}\right\|_{1}\\
         &\leq \left\|t_{m^{n}, x^{n}, y^{n}}-p\cdot t_{x^{n}, y^{n}}\right\|_{1}+\left\|p\cdot t_{x^{n}, y^{n}}-p\cdot \widetilde{t}\right\|_{1}\\
        &\leq \delta' +\sum_{(m, x, y) \in \mathcal{M}\times\mathcal{X}\times\mathcal{Y}}\left|p(m|x)\cdot t_{x^{n}, y^{n}}(x, y)-p(m|x)\cdot \widetilde{t}(x, y) \right|\\
        &=\delta' +\sum_{(m, x, y) \in \mathcal{M}\times\mathcal{X}\times\mathcal{Y}}\left|p(m|x)\left(t_{x^{n}, y^{n}}(x, y)-\widetilde{t}(x, y)\right) \right|.\\
        \intertext{Since $\forall (m, x) \in \mathcal{M}\times\mathcal{X}$ it must hold that $p(m|x)>0$,}
        &= \delta' +\sum_{(m, x, y) \in \mathcal{M}\times\mathcal{X}\times\mathcal{Y}}p(m|x)\left|t_{x^{n}, y^{n}}(x, y)-\widetilde{t}(x, y) \right|\\
        &= \delta' +\sum_{(x, y) \in \mathcal{X}\times\mathcal{Y}}\left|t_{x^{n}, y^{n}}(x, y)-\widetilde{t}(x, y) \right|\sum_{m \in \mathcal{M}}p(m|x)\\
        &\leq \delta' +\sum_{(x, y) \in \mathcal{X}\times\mathcal{Y}}\left|t_{x^{n}, y^{n}}(x, y)-\widetilde{t}(x, y) \right|\\
        &\leq \delta' +\delta.
        \end{align}
    Since this holds for every $(m^{n}, x^{n}, y^{n}) \in T^{M^{n}, X^{n}, Y^{n}}_{(\widetilde{t}, \delta), (p, \delta')} $, we have that
    \begin{equation}
        T^{M^{n}, X^{n}, Y^{n}}_{(\widetilde{t}, \delta), (p, \delta')} \subseteq T^{M^{n}, X^{n}, Y^{n}}_{(p \cdot \widetilde{t}, \delta'+\delta)}
    \end{equation}
\end{proof}
            \begin{proof}[Proof of Claim~\ref{clm:derandom_sw_2}]
                    \label{proof:derandom_sw_2}
                        Let us pick $s$ independent random strings to use for the unbounded randomness protocol. Let $\delta''>0$, then for any input $(x^{n}, y^{n})$, the probability that a $\frac{3\cdot2^{-n\delta'}}{4}+\delta''$ fraction of the $s$ strings lead to $\hat{x}^{n} \neq x^{n}$ is at most $2^{-\Omega(s\delta''^{2})}$ by the Chernoff bound. Letting
                    \begin{equation}
                        s = O\left(\frac{n\log\left(|\mathcal{X}|\cdot|\mathcal{Y}| \right)}{\delta''^{2}} \right),
                    \end{equation}
                    we find that this probability is smaller than $\left(|\mathcal{X}|^{n}\cdot|\mathcal{Y}|^{n}\right)^{-1}$, so by the union bound, the probability that more than a $\frac{3\cdot2^{-n\delta'}}{4}+\delta''$ fraction of the $s$ strings gives the wrong answer for any input $(x^{n}, y^{n}) \in \mathcal{X}^{n} \times \mathcal{Y}^{n}$ is less than $1$. 

                    Our bounded randomness protocol is then as follows. Sample one of these $s$ strings at random, and then run the corresponding unbounded randomness protocols. Let $E_{good}$ be the event that the sampled string leads to $\hat{x}^{n} = x^{n}$ for input $(x^{n}, y^{n})$, then 
                    \begin{equation}
                        \Pr[E_{good}] \geq 1-\frac{3\cdot2^{-n\delta'}}{4}-\delta''.
                    \end{equation}
                    Picking $\delta'' = \frac{2^{-n\delta'}}{4}$ leads to $\Pr[E_{good}] \geq 1 - 2^{-n\delta'}$, and the number of bits of shared randomness 
                    \begin{equation}
                        \log(s) = \log(n) +\log\log\left(|\mathcal{X}|\cdot|\mathcal{Y}|\right)+2n\delta'+O(1)
                    \end{equation}
                     as desired.
                    \end{proof}
            \begin{proof}[Proof of Theorem~\ref{thm:rst_1}]
            \label{proof:rst_1}
                An error will occur in the protocol if the message $m^{n}$, generated by Alice through her local use of the channel, is not in $T^{M^{n}|x^{n}\backslash Y^{n}}_{(\widetilde{t}, \delta), (p,  \delta')}$ $(E_{1})$, if there is no message in the subset $\mathcal{L}_{r}$ that has the correct joint type with Alice's input $(E_{2})$, in which case Alice will have no messages to choose from, if the message $m^{n}_{r}$ chosen by Alice from the subset $\mathcal{L}_{r}$ is not in $T^{M^{n}|y^{n}\backslash X^{n}}_{(\widetilde{t}, \delta), (p,  \delta')}$ $(E_{3})$, or if $\left|\mathcal{L}_{r, c} \cap T^{M^{n}|y^{n}\backslash X^{n}}_{(\widetilde{t}, \delta), (p,  \delta')}\right| > 1$ $(E_{4})$, in which case Bob will not be able to determine which message Alice chose.
                
                For $\Pr\left[E_{1}\right]$ and $\Pr\left[E_{3}\right]$, let $F$ be the event that $m^{n} \in T^{M^{n}|x^{n}, y^{n}}_{t, (p,\delta')}$. Proposition~\ref{prop:unit_prob} gives the following bound,
                \begin{equation}
                    \Pr\left[F\right] \geq 1-2^{-n\delta'''}
                \end{equation}
                where $\delta''' = \left(\frac{\delta'^{2}}{2\ln 2}-\frac{2}{n}|\mathcal{M}|\cdot|\mathcal{X}|\cdot|\mathcal{Y}|\log(n+1)\right)$. If $m^{n} \in T^{M^{n}|x^{n}, y^{n}}_{t, (p,\delta')}$, then since $\left\|t_{x^{n}, y^{n}}-\widetilde{t}\right\|_{1} \leq \delta$, the tuple $(m^{n}, x^{n}, y^{n})$ must be in the jointly typical set $T^{M^{n}, X^{n}, Y^{n}}_{(\widetilde{t}, \delta), (p, \delta')}$. It then follows that there exists a sequence in $\mathcal{Y}^{n}$, specifically Bob's sequence $y^{n}$, such that $m^{n} \in T^{M^{n}|x^{n}\backslash Y^{n}}_{(\widetilde{t}, \delta), (p,  \delta')}$. Since $m^{n}_{r}$ and $m^{n}$ have the same joint type with $x^{n}$, it also holds that the tuple $(m^{n}_{r}, x^{n}, y^{n})$ is in the jointly typical set $T^{M^{n}, X^{n}, Y^{n}}_{(\widetilde{t}, \delta), (p, \delta')}$. It therefore follows that $m^{n}_{r} \in T^{M^{n}|y^{n}\backslash X^{n}}_{(\widetilde{t}, \delta), (p,  \delta')}$. We then have that $\Pr[\lnot E_{1}] \geq \Pr[F]$ and $\Pr[\lnot E_{3}] \geq \Pr[F]$, and therefore
                \begin{align}
                    \Pr[E_{1}] &=1-\Pr[\lnot E_{1}]\\
                    &\leq 2^{-n\delta'''},
                \end{align}
                and similarly
                \begin{equation}
                    \Pr[E_{3}] \leq 2^{-n\delta'''}.
                \end{equation}
                For $\Pr\left[E_{2}\right]$, when the list $\mathcal{L}_{t}$ is reduced to the messages whose position in the list correspond to the $nR$ bits specified by the shared randomness, the number of messages in the list is reduced by a factor of $\left(\frac{1}{2} \right)^{nR}$ to obtain the new list $\mathcal{L}_{r}$, so $\left|\mathcal{L}_{r}\right| = \frac{\left|\mathcal{L}_{t}\right|}{2^{nR}} = \frac{\left|T^{M^{n}}_{t}\right|}{2^{nR}}$. Using the notation of Lemma~\ref{lem:set_inequality} (Hoeffding's inequality), let 
                \begin{equation}
                   U = T^{M^{n}|x^{n}}_{t}, u = \left|T^{M^{n}|x^{n}}_{t}\right|
                \end{equation} 
                \begin{equation}
                    V = \mathcal{L}_{r}, v = \frac{\left|\mathcal{L}_{t} \right|}{2^{nR}}= \frac{\left|T^{M^{n}}_{t} \right|}{2^{nR}}
                \end{equation}
                \begin{equation}
                    Q = \mathcal{L}_{t} = T^{M^{n}}_{t}, q = \left|T^{M^{n}}_{t}\right|
                \end{equation} 
                and 
                \begin{equation}
                    \mu = \frac{\left|T^{M^{n}|x^{n}}_{t}\right| \cdot \frac{\left|T^{M^{n}}_{t} \right|}{2^{nR}}}{\left|T^{M^{n}}_{t}\right|} = \frac{\left|T^{M^{n}|x^{n}}_{t}\right|}{2^{nR}}.
                \end{equation}
                For $\delta''<\left(1-\frac{1}{\mu}\right)$ we have $(1-\delta'')\mu>1$, then
                \begin{align}
                    \Pr\left[\left| \mathcal{L}_{r} \cap T^{M^{n}|x^{n}}_{t} \right| \leq (1-\delta'')\mu\right] &= \Pr\left[\left| \mathcal{L}_{r} \cap T^{M^{n}|x^{n}}_{t} \right| < 1\right]+\Pr\left[1\leq\left| \mathcal{L}_{r} \cap T^{M^{n}|x^{n}}_{t} \right| \leq (1-\delta'')\mu\right]\\
                    &\geq \Pr\left[\left| \mathcal{L}_{r} \cap T^{M^{n}|x^{n}}_{t} \right| < 1\right]
                \end{align}
                
                Applying Lemma~\ref{lem:set_inequality} gives
                \begin{align}
                   \Pr\left[ \left| \mathcal{L}_{r} \cap T^{M^{n}|x^{n}}_{t} \right|  < 1\right] &\leq e^{-\frac{\frac{\left|T^{M^{n}|x^{n}}_{t}\right|}{2^{nR}} \delta''^{2}}{2}}\\
                   &= 2^{-\frac{\left|T^{M^{n}|x^{n}}_{t}\right|}{2^{nR+1}} \delta''^{2}\cdot \log(e)}\\
                    &\leq 2^{-\frac{2^{nH\left(M| X\right)_{t}-|\mathcal{M}| \times |\mathcal{X}|\log\left(n+1\right)+\log\log(e)}}{2^{nR+1}}\delta''^{2}}.
                \end{align}
                This gives the following upper bound on $R$
                \begin{equation}
                    R \leq H\left(M| X\right)_{t}-\frac{1}{n}|\mathcal{M}| \times |\mathcal{X}|\log\left(n+1\right)+\frac{1}{n}\log\log(e)-\frac{1}{n}-\frac{1}{n}\log(n)
                \end{equation}
                in which case
                \begin{equation}
                   \Pr\left[ \left| \mathcal{L}_{r} \cap T^{M^{n}|x^{n}}_{t} \right|  < 1\right] \leq 2^{-n\delta''^{2}}.
                \end{equation}
                Conditioning on $\lnot E_{1}$, it holds that
                \begin{align}
                    \left\|t_{m^{n}, x^{n}, y^{n}}-p_{M|X}\cdot t_{x^{n}, y^{n}}\right\|_{1} &\leq \delta'\\
                    \sum_{y \in \mathcal{Y}}\left\|t_{m^{n}, x^{n}, y^{n}}-p_{M|X}\cdot t_{x^{n}, y^{n}}\right\|_{1} &\leq \left |\mathcal{Y}\right|\delta'\\
                    \left\|t_{m^{n}, x^{n}}-p_{M|X}\cdot t_{x^{n}}\right\|_{1} &\leq \left |\mathcal{Y}\right|\delta'.
                \end{align}
                Applying Lemma~\ref{lem:cont_bounds} (Alicki-Fannes inequality), Alice and Bob can choose $R$ to be
                \begin{equation}
                    R = H\left(M| X\right)_{t\cdot p}-\gamma(\left|\mathcal{M}\right|,  \left |\mathcal{Y}\right|\cdot\delta')-\frac{1}{n}|\mathcal{M}| \times |\mathcal{X}|\log\left(n+1\right)+\frac{1}{n}\log\log(e)-\frac{1}{n}-\frac{1}{n}\log(n).
                \end{equation}
                Since we are guaranteed that $(x^{n}, y^{n}) \in T^{X^{n}, Y^{n}}_{\widetilde{t}, \delta}$, it holds that
                \begin{equation}
                    \left\|t_{x^{n},y^{n}}-\widetilde{t}\right\|_{1} \leq \delta,
                \end{equation}
                which implies that
                \begin{align}
                    \left\|p_{M|X}\cdot t_{x^{n}, y^{n}}-p_{M|X}\cdot\widetilde{t}\right\|_{1} &= \sum_{(m, x, y)\in\mathcal{M}\times\mathcal{X}\times\mathcal{Y}} \left| p_{M|X}(m|x)\cdot t_{x^{n}, y^{n}}(x, y)-p_{M|X}(m|x)\cdot\widetilde{t}(x, y)\right|\\
                    &= \sum_{(m, x, y)\in\mathcal{M}\times\mathcal{X}\times\mathcal{Y}} p_{M|X}(m|x)\left| t_{x^{n}, y^{n}}(x, y)-\widetilde{t}(x, y)\right|\\
                    &= \sum_{(x, y)\in\mathcal{X}\times\mathcal{Y}}\left| t_{x^{n}, y^{n}}(x, y)-\widetilde{t}(x, y)\right|\\
                    &\leq \delta.
                \end{align}
                Further, using the property that $H(M|X) = H(M|X, Y)$ when calculated for the joint distribution $p_{M|X}\cdot \widetilde{t}$, we obtain that Alice and Bob can ultimately choose
                \begin{align}
                \label{eqn:r_final}
                    R &= H\left(M|X, Y\right)_{\widetilde{t}\cdot p}-\gamma(\left|\mathcal{M}\right|, \delta)-\gamma(\left|\mathcal{M}\right|,  \left |\mathcal{Y}\right|\cdot\delta')-\frac{1}{n}|\mathcal{M}| \cdot |\mathcal{X}|\log\left(n+1\right)+\frac{1}{n}\log\log(e)-\frac{1}{n}-\frac{1}{n}\log(n)\\
                    &\leq H\left(M| X, Y\right)_{t\cdot p}+\frac{1}{n}\log\log(e).
                \end{align}
                To determine $ \Pr\left[E_{4}\right]$, we use the fact that the positions in the list $\mathcal{L}$ are assigned uniformly at random
                \begin{align}
                    \Pr\left[E_{4}\right] &= \Pr\left[\exists \widetilde{m^{n}} \neq m^{n}_{r}\in T^{M^{n}|y^{n}\backslash X^{n}}_{(\widetilde{t}, \delta), (p,  \delta')}: \widetilde{m^{n}} \in \mathcal{L}_{r, c}\right]
                    \\&\leq \sum_{\widetilde{m^{n}} \neq m^{n}_{r} \in T^{M^{n}|y^{n}\backslash X^{n}}_{(\widetilde{t}, \delta), (p,  \delta')}}\Pr\left[\widetilde{m^{n}} \in \mathcal{L}_{r, c}\right]
                    \\&= \left|T^{M^{n}|y^{n}\backslash X^{n}}_{(\widetilde{t}, \delta), (p,  \delta')}\right| \frac{1}{2^{n(R+C)}}.
                \end{align}
                Using Property~\ref{prop:card_cond_bs_set} we obtain the following bound
                \begin{equation}
                    C+R > H\left(M| Y\right)_{\widetilde{t}\cdot p}+\gamma(|\mathcal{M}|, \left|\mathcal{X}\right|\cdot\delta')+\frac{2}{n}|\mathcal{M}|\cdot |\mathcal{Y}|\log\left(n+1\right)+\delta.
                \end{equation}
                Alice and Bob can choose $R$ as in (\ref{eqn:r_final}) and $C$ as 
                \begin{align}
                    C &= I\left(M;X|Y\right)_{\widetilde{t}\cdot p}+\gamma(|\mathcal{M}|, \left|\mathcal{X}\right|\cdot\delta')+\frac{2}{n}|\mathcal{M}|\cdot |\mathcal{Y}|\log\left(n+1\right)+\delta+\gamma(\left|\mathcal{M}\right|, \delta)+\gamma(\left|\mathcal{M}\right|, |\mathcal{Y}|\cdot\delta')\\
                    &+\frac{1}{n}|\mathcal{M}| \cdot |\mathcal{X}|\log\left(n+1\right)-\frac{1}{n}\log\log(e)+\frac{1}{n}+\frac{1}{n}\log(n)\\
                    \intertext{and applying the Alicki-Fannes inequality to define the mutual information in terms of the true joint type $t_{x^{n}, y^{n}}$ gives the upperbound}
                    &\leq I\left(M;X|Y\right)_{t\cdot p}+\gamma(|\mathcal{M}|, \left|\mathcal{X}\right|\cdot\delta')+3\gamma(\left|\mathcal{M}\right|, \delta)+\gamma(\left|\mathcal{M}\right|, |\mathcal{Y}|\cdot\delta')
                    +\frac{3}{n}|\mathcal{M}| \cdot|\mathcal{X}|\cdot |\mathcal{Y}|\log\left(n+1\right)+\delta\\&+\frac{1}{n}+\frac{1}{n}\log(n).
                    \end{align}
                The resulting probability of error is then
                \begin{equation}
                     \Pr\left[E_{4}\right] \leq 2^{-n\delta}.
                \end{equation}
                
                For any input $(x^{n}, y^{n})$, the statistical distance between the output distribution of the channel and that of the simulation is
                \begin{align}
                    &\sum_{\widetilde{m^{n}} \in \mathcal{M}^{n}}\left|p^{n} \left(\widetilde{m^{n}}|x^{n}, y^{n} \right) - \widetilde{p^{n}} \left(\widetilde{m^{n}}|x^{n}, y^{n} \right)\right|\\
                    &= \sum_{\widetilde{m^{n}} \not\in T^{M^{n}|x^{n}, y^{n}}_{t,(p, \delta')}}\left|p^{n} \left(\widetilde{m^{n}}|x^{n} \right) - \widetilde{p^{n}} \left(\widetilde{m^{n}}|x^{n}, y^{n} \right)\right| + \sum_{\widetilde{m^{n}} \in T^{M^{n}|x^{n}, y^{n}}_{t,(p, \delta')}}\left|p^{n} \left(\widetilde{m^{n}}|x^{n} \right) - \widetilde{p^{n}} \left(\widetilde{m^{n}}|x^{n}, y^{n} \right)\right|.
                \end{align}
                For the first sum, a message not in $T^{M^{n}|x^{n}, y^{n}}_{t,(p, \delta')}$ will result in an error. To prove this, first note that in this case $T^{M^{n}|x^{n}, y^{n}}_{t,(p, \delta')}=T^{M^{n}|x^{n}}_{t,(p, \delta')}$. If the message generated by Alice through her local use of the channel $m^{n}$ is not in $T^{M^{n}|x^{n}}_{t,(p, \delta')}$, then a tuple containing $m^{n}$ and $x^{n}$ will not appear in $T^{M^{n}, X^{n}, Y^{n}}_{(\widetilde{t}, \delta), (p, \delta')}$, and subsequently $m^{n}$ will not appear in the set $T^{M^{n}|x^{n}\backslash Y^{n}}_{(\widetilde{t}, \delta), (p,  \delta')}$ which will result in an error in the protocol. In this case, we have $\widetilde{p^{n}} \left(\widetilde{m^{n}}|x^{n} \right) = 0$. By Proposition~\ref{prop:unit_prob}, the first sum is bound by
                \begin{align}
                    \sum_{\widetilde{m^{n}} \not\in T^{M^{n}|x^{n}, y^{n}}_{t,(p, \delta')}}\left|p^{n} \left(\widetilde{m^{n}}|x^{n} \right)\right| &= 1- \sum_{\widetilde{m^{n}} \in T^{M^{n}|x^{n}, y^{n}}_{t,(p, \delta')}}p^{n} \left(\widetilde{m^{n}}|x^{n} \right)\\
                    &\leq 2^{-n\delta'''}.
                \end{align}
                Recall
                \begin{equation}
                    T^{M^{n}|x^{n}, y^{n}}_{t,(p, \delta')}= \bigcup_{\substack{\overline{t} \in \mathcal{T}^{\mathcal{M}^{n}, \mathcal{X}^{n}, \mathcal{Y}^{n}}\\ \|\overline{t} -p\cdot t\|_{1} \leq \delta'}}T^{M^{n} |x^{n}, y^{n}}_{\overline{t}}.
                \end{equation}
                The second sum can be decomposed as
                \begin{align}
                    &= \sum_{\substack{\overline{t} \in \mathcal{T}^{\mathcal{M}^{n}, \mathcal{X}^{n}}\\ \|\overline{t} -p\cdot t\|_{1} \leq \delta'}}\sum_{r=1}^{2^{nR}}\sum_{\widetilde{m^{n}} \in T_{\overline{t}}^{M^{n} |x^{n}} \cap \mathcal{L}_{r}}\left|p^{n} \left(\widetilde{m^{n}}|x^{n} \right) - \widetilde{p^{n}} \left(\widetilde{m^{n}}|x^{n}, y^{n} \right)\right|\\
                    &= \sum_{\substack{\overline{t} \in \mathcal{T}^{\mathcal{M}^{n}, \mathcal{X}^{n}}\\ \|\overline{t} -p\cdot t\|_{1} \leq \delta'}} \Pr\left[t\left(m^{n}, x^{n}\right) = \overline{t}\right]  \sum_{r=1}^{2^{nR}}\sum_{\widetilde{m^{n}} \in T_{\overline{t}}^{M^{n} |x^{n}} \cap \mathcal{L}_{r}}\left|\frac{1}{\left| T_{\overline{t}}^{M^{n} |x^{n}} \right|} - \frac{1}{2^{nR}} \cdot\frac{1}{\left|T_{\overline{t}}^{M^{n} |x^{n}} \cap \mathcal{L}_{r} \right|}\right|\\
                    &= \sum_{\substack{\overline{t} \in \mathcal{T}^{\mathcal{M}^{n}, \mathcal{X}^{n}}\\ \|\overline{t} -p\cdot t\|_{1} \leq \delta'}} \Pr\left[t\left(m^{n}, x^{n}\right) = \overline{t}\right]  \sum_{r=1}^{2^{nR}}\left|T_{\overline{t}}^{M^{n} |x^{n}} \cap \mathcal{L}_{r} \right|\left|\frac{1}{\left| T_{\overline{t}}^{M^{n} |x^{n}} \right|} - \frac{1}{2^{nR}} \cdot\frac{1}{\left|T_{\overline{t}}^{M^{n} |x^{n}} \cap \mathcal{L}_{r} \right|}\right|\\
                    &= \sum_{\substack{\overline{t} \in \mathcal{T}^{\mathcal{M}^{n}, \mathcal{X}^{n}}\\ \|\overline{t} -p\cdot t\|_{1} \leq \delta'}} \Pr\left[t\left(m^{n}, x^{n}\right) = \overline{t}\right] \sum_{r=1}^{2^{nR}}\frac{1}{\left| T_{\overline{t}}^{M^{n} |x^{n}} \right|}\left|\left|T_{\overline{t}}^{M^{n} |x^{n}} \cap \mathcal{L}_{r} \right| - \frac{\left| T_{\overline{t}}^{M^{n} |x^{n}} \right|}{2^{nR}}\right|.
                \intertext{We can use Lemma~\ref{lem:set_inequality} (Hoeffding's inequality) to show that}
                    &\Pr\left [ \left|\left|T^{M^{n}|x^{n}}_{\overline{t}} \cap \mathcal{L}_{r} \right| - \frac{\left| T^{M^{n}|x^{n}}_{\overline{t}} \right|}{2^{nR}}\right| < \epsilon \frac{\left| T^{M^{n}|x^{n}}_{\overline{t}} \right|}{2^{nR}} \right ] \geq 1 - 2e^{-\frac{\frac{\left| T^{M^{n}|x^{n}}_{\overline{t}} \right|}{2^{nR}}\epsilon^{2}}{2}}.
                \end{align}
                By choosing $R$ as in (\ref{eqn:r_final}), the probability is at least $1-2^{-n\epsilon^{2}}$, and with this we can upper bound the sum by $\epsilon$.
                The total variation distance of the simulation is therefore bounded by
                \begin{align} \label{eqn:rst_tot_var_dist}
                    \left \|p^{n} \left(\mathbf{m^{n}}|x^{n} \right) -  \widetilde{p^{n}} \left(\mathbf{m^{n}}|x^{n} \right) \right\|_{1} \leq 2^{-n\delta'''} + \delta''.
                \end{align}
                By the union bound, the probability of success of the protocol is at least
                \begin{align}
                     &1 - 2\cdot2^{-n\delta'''}-2^{-n\delta''^{2}}-2^{-n\delta}\\
                     &\geq 1 - 2^{-n\delta_{\textnormal{min}}^{2}+2}.
                \end{align}
            \end{proof}   
            \begin{proof}[Proof of Claim~\ref{clm:derandom_rst}]
            \label{proof:derandom_rst}
                    Let us pick $s$ independent random strings to use for the unbounded randomness protocol, where $n(H\left(M|X, Y\right)_{\widetilde{t}\cdot p}+\frac{1}{n}\log\log(e))$ bits of true shared randomness will still be used. Let $\overline{\delta}>0$, then for any input $(x^{n}, y^{n})$, the probability that a $2^{-n\delta_{\textnormal{min}}^{2}+2}+\overline{\delta}$ fraction of the $s$ strings lead to an error in the protocol is at most $2^{-\Omega(s\overline{\delta}^{2})}$ by the Chernoff bound. Letting
                    \begin{equation}
                        s = O\left(\frac{n\log\left(|\mathcal{X}|\cdot|\mathcal{Y}| \right)}{\overline{\delta}^{2}} \right),
                    \end{equation}
                    we find that this probability is smaller than $\left(|\mathcal{X}|^{n}\cdot|\mathcal{Y}|^{n}\right)^{-1}$, so by the union bound, the probability that more than a $2^{-n\delta_{\textnormal{min}}^{2}+2}+\overline{\delta}$ fraction of the $s$ strings gives the wrong answer for any input $(x^{n}, y^{n}) \in \mathcal{X}^{n} \times \mathcal{Y}^{n}$ is less than $1$. 

                    Our bounded randomness protocol is then as follows. Sample one of these $s$ strings at random, and then run the corresponding unbounded randomness protocol. Let $E_{good}$ be the event that the sampled string does not lead to an error in the protocol on input $(x^{n}, y^{n})$, then 
                    \begin{align}
                        \Pr[E_{good}] &\geq 1-2^{-n\cdot\delta_{\textnormal{min}}^{2}+2}-\overline{\delta}\\
                        &= 1-\frac{2^{-n\cdot\delta_{\textnormal{min}}^{2}+3}}{2}-\overline{\delta}
                    \end{align}
                    Picking $\overline{\delta} = \frac{2^{-n\cdot\delta_{\textnormal{min}}^{2}+3}}{4}$ leads to $\Pr[E_{good}] \geq 1 - 2^{-n\cdot\delta_{\textnormal{min}}^{2}+3}$.
                \end{proof}

                \begin{proof}[Proof of Claim~\ref{clm:derandom_rst_2}]
                \label{proof:derandom_rst_2}
                        Let us pick $s$ independent random strings to use for the unbounded randomness protocol, where $n(H\left(M|X, Y\right)_{t\cdot p}+\frac{1}{n}\log\log(e))$ bits of true shared randomness will still be used. Let $\overline{\delta}>0$, then for any input $(x^{n}, y^{n})$, the probability that a $5\cdot2^{-n\cdot\delta_{\textnormal{min}}^{2}}+\overline{\delta}$ fraction of the $s$ strings lead to an error in the protocol is at most $2^{-\Omega(s\overline{\delta}^{2})}$ by the Chernoff bound. Letting
                    \begin{equation}
                        s = O\left(\frac{n\log\left(|\mathcal{X}|\cdot|\mathcal{Y}| \right)}{\overline{\delta}^{2}} \right),
                    \end{equation}
                    we find that this probability is smaller than $\left(|\mathcal{X}|^{n}\cdot|\mathcal{Y}|^{n}\right)^{-1}$, so by the union bound, the probability that more than a $5\cdot2^{-n\cdot\delta_{\textnormal{min}}^{2}}+\overline{\delta}$ fraction of the $s$ strings gives the wrong answer for any input $(x^{n}, y^{n}) \in \mathcal{X}^{n} \times \mathcal{Y}^{n}$ is less than $1$. 

                    Our bounded randomness protocol is then as follows. Sample one of these $s$ strings at random, and then run the corresponding unbounded randomness protocols. Let $E_{good}$ be the event that the sampled string leads to no error in the protocol, then 
                    \begin{align}
                        \Pr[E_{good}] &\geq 1-5\cdot2^{-n\cdot\delta_{\textnormal{min}}^{2}}-\overline{\delta}\\
                        &= 1-\frac{3}{4}\cdot\frac{4}{3}5\cdot2^{-n\cdot\delta_{\textnormal{min}}^{2}}-\overline{\delta}\\
                        &= 1-\frac{3}{4}\cdot2^{-n\cdot\delta_{\textnormal{min}}^{2}+\log_{2}\left(\frac{20}{3} \right)}-\overline{\delta}\\
                        &\geq 1-\frac{3}{4}\cdot2^{-n\cdot\delta_{\textnormal{min}}^{2}+3}-\overline{\delta}.
                    \end{align}
                    Picking $\overline{\delta} = \frac{2^{-n\cdot\delta_{\textnormal{min}}^{2}+3}}{4}$ leads to $\Pr[E_{good}] \geq 1 - 2^{-n\cdot\delta_{\textnormal{min}}^{2}+3}$.
                    \end{proof}
            \begin{proof}[Proof of Theorem~\ref{thm:int_2}]
            \label{proof:int_2} 
            The communication cost for the $(i-1)^{th}$ iteration of step $2$, which generates the $i^{th}$ message of the protocol (assume $i$ is odd) is bound by
            \begin{equation}
                n\left(I\left(M;X|Y\right)_{p_{i}\cdot p_{i-1}\cdot...\cdot p_{1}\cdot t}+\eta_{2}(n, \delta_{\textnormal{max}, i})\right)
            \end{equation}
            where $\delta_{\textnormal{max}, i} = \max\left(\delta+(i-1)\delta', \delta', \delta'' \right)$. Note that the parameter corresponding to $\delta+(i-1)\delta'$ is the only one that needs to be updated from one round to the next. The parameters $\delta', \delta'', \delta_{s}$ and consequently $\delta'''$ remain unchanged.

            The total communication cost of simulating all $j$ rounds is obtained by summing the cost over all $j$ rounds. Note that line (\ref{eqn:xtra_line}) is equal to $0$.
            \begin{align}
                C &\leq I(M_{1};X|Y)_{p_{1}\cdot t_{x^{n}, y^{n}}}\\
                &+\sum_{i = 3, 5,...} I(M_{i};X|YM_{1}...M_{i-1})_{p_{i} \cdot p_{i-1}\cdot ...\cdot p_{1} \cdot t_{x^{n}, y^{n}}}
                +\sum_{i = 2, 4, 6,...} I(M_{i};Y|XM_{1}...M_{i-1})_{p_{i} \cdot p_{i-1}\cdot ...\cdot p_{1} \cdot t_{x^{n}, y^{n}}} \\\label{eqn:xtra_line}
                &+\sum_{i=1, 2,..., j} I(M_{i};X|YM_{1}...M_{i-1})_{p_{i} \cdot p_{i-1}\cdot ...\cdot p_{1} \cdot t_{x^{n}, y^{n}}}+ \sum_{i = 3, 5,...}I(M_{i};Y|XM_{1}...M_{i-1})_{p_{i} \cdot p_{i-1}\cdot ...\cdot p_{1} \cdot t_{x^{n}, y^{n}}}\\
                &+\eta_{3}(n, \delta_{\textnormal{max}, j}, j)+k\log\left(\left|\mathcal{X}\right|\cdot\left|\mathcal{Y}\right|\right)\\
                &\leq I(M_{1}...M_{j};X|Y)_{p_{i} \cdot p_{i-1}\cdot ...\cdot p_{1} \cdot t_{x^{n}, y^{n}}}+I(M_{1}...M_{j};Y|X)_{p_{i} \cdot p_{i-1}\cdot ...\cdot p_{1} \cdot t_{x^{n}, y^{n}}}+\eta_{3}(n, \delta_{\textnormal{max}}, j)\\
                &+k\log\left(\left|\mathcal{X}\right|\cdot\left|\mathcal{Y}\right|\right)
            \end{align}
            where
            \begin{align}
                \eta_{3}(n, \delta_{\textnormal{max}, j}, j) = \sum_{i=1, 2,..., j} 5\gamma(|\mathcal{M}_{i}|, \left|\mathcal{X}\right|\cdot\left|\mathcal{Y}\right|\cdot\delta_{\textnormal{max}, i})+\frac{4}{n}|\mathcal{M}_{\leq i}| \cdot|\mathcal{X}|\cdot |\mathcal{Y}|\log\left(n+1\right)+\frac{2}{n}\log(n)\\+3\delta_{\textnormal{max}, i}+O\left(\frac{1}{n}\right)\\
                \leq 5j\gamma(|\mathcal{M}_{max}|, \left|\mathcal{X}\right|\cdot\left|\mathcal{Y}\right|\cdot\delta_{\textnormal{max}, j})+\frac{4j}{n}|\mathcal{M}_{\leq j}| \cdot|\mathcal{X}|\cdot |\mathcal{Y}|\log\left(n+1\right)+\frac{2j}{n}\log(n)\\+3j\delta_{\textnormal{max}, j}+O\left(\frac{j}{n}\right)
                \end{align}
                where for every $i \leq j$, it holds that
                \begin{equation}
                    \delta_{\textnormal{max}, j}\geq \delta_{\textnormal{max}, i}
                \end{equation}
                and we have defined
                \begin{equation}
                |\mathcal{M}_{max}| = \max_{i=1, 2,...,j} |\mathcal{M}_{i}|.
                \end{equation}

            For all $j$ rounds, we can show that the total variation distance between the true and simulated channels is bound by 
            \begin{align}
                \big\|\widetilde{p^{n}}&(m_{1}^{n}, ..., m_{j}^{n}|x^{n}, y^{n})-p^{n}(m_{1}^{n}, ..., m_{j}^{n}|x^{n}, y^{n})\big\|_{1}\\
                &=\big\|\widetilde{p_{j}^{n}}(m_{j}^{n}|x^{n}, y^{n}, m_{<j}^{n})\cdot\widetilde{p_{<j}^{n}}(m_{<j}^{n}|x^{n}, y^{n}) -p_{j}^{n}(m_{j}^{n}|x^{n}, y^{n}, m_{<j}^{n})\cdot p_{<j}^{n}(m_{<j}^{n}|x^{n}, y^{n})
                \\
                &\pm \widetilde{p_{j}^{n}}(m_{j}^{n}|x^{n}, y^{n}, m_{<j}^{n})\cdot p_{<j}^{n}(m_{<j}^{n}|x^{n}, y^{n})\big\|_{1}\\
                &= \big\|\widetilde{p_{j}^{n}}(m_{j}^{n}|x^{n}, y^{n}, m_{<j}^{n})\Big(\widetilde{p_{<j}^{n}}(m_{<j}^{n}|x^{n}, y^{n})- p_{<j}^{n}(m_{<j}^{n}|x^{n}, y^{n})\Big)\\
                &+\left(\widetilde{p_{j}^{n}}(m_{j}^{n}|x^{n}, y^{n}, m_{<j}^{n})-p_{j}^{n}(m_{j}^{n}|x^{n}, y^{n}, m_{<j}^{n}) \right)p_{<j}^{n}(m_{<j}^{n}|x^{n}, y^{n})\big\|_{1}\\
                &\leq \left\|\widetilde{p_{<j}^{n}}(m_{<j}^{n}|x^{n}, y^{n})- p_{<j}^{n}(m_{<j}^{n}|x^{n}, y^{n})\right\|_{1}+\mathbb{E}_{m^{n}_{<j}}\left\|\widetilde{p_{j}^{n}}(m_{j}^{n}|x^{n}, y^{n}, m_{<j}^{n})-p_{j}^{n}(m_{j}^{n}|x^{n}, y^{n}, m_{<j}^{n}) \right\|_{1}\\
               &\leq \left\|\widetilde{p_{1}^{n}}(m_{1}^{n}|x^{n}, y^{n})-p_{1}^{n}(m_{1}^{n}|x^{n}, y^{n})\right\|_{1} + \mathbb{E}_{m^{n}_{1}}\left \|\widetilde{p_{2}^{n}}(m_{1}^{n}|x^{n}, y^{n}, m^{n}_{1})-p_{2}^{n}(m_{1}^{n}|x^{n}, y^{n}, m^{n}_{1})\right\|_{1}+...\\
               &...+\mathbb{E}_{m^{n}_{<j}}\left\|\widetilde{p_{j}^{n}}(m_{j}^{n}|x^{n}, y^{n}, m_{<j}^{n})-p_{j}^{n}(m_{j}^{n}|x^{n}, y^{n}, m_{<j}^{n})\right\|_{1}
            \end{align}
            \begin{equation} \label{eqn:rst_int_tot_var_dist}
                \leq j\left(2^{-n\delta'''} + \delta''\right).
            \end{equation}
            \end{proof}
\end{document}